\documentclass[a4paper,USenglish]{lipics-v2021}
  
\usepackage{libertine}

\usepackage{datetime}
\usepackage{url}
\usepackage{paralist}

\usepackage{algorithm,algpseudocode}
\usepackage{multicol}

\usepackage[font=small,labelfont=bf,justification=centering]{caption}

\usepackage{amsmath,amsthm}

\theoremstyle{definition}

\theoremstyle{definition}

\usepackage{graphicx}

\usepackage{arydshln}

\usepackage{multirow}

\usepackage{threeparttable}

\usepackage{cleveref}
\crefname{section}{\S}{\S\S}
\Crefname{section}{\S}{\S\S}
\crefformat{section}{\S#2#1#3}

\Crefname{assumption}{Assumption}{Assumptions}

\Crefname{invariant}{Invariant}{Invariants}

\Crefname{observation}{Observation}{Observations}

\usepackage{algorithmicx}
\usepackage{algorithm} 
\usepackage{algpseudocode}

\usepackage{soul}

\usepackage{xspace}

\newcommand{\name}{\textsc{DARE}\xspace}

\newcommand{\darestark}{\textsc{DARE-Stark}\xspace}
\newcommand{\sync}{\textsc{Sync}\xspace}

\newcommand{\spread}{\textsc{Disperser}\xspace}
\newcommand{\q}{\textsc{Agreement}\xspace}
\newcommand{\rec}{\textsc{Retriever}\xspace}

\newcommand{\bcastbasedsync}{\textsc{BcastBasedSync}\xspace}

\usepackage{xcolor}
\usepackage{framed}
\definecolor{lightgray}{gray}{0.90}
\renewenvironment{leftbar}[1][\hsize]
{%
\MakeFramed{\hsize#1\advance\hsize-\width\FrameRestore}%
}
{\endMakeFramed}

\algsetblockdefx[Class]{Class}{EndClass}{}{}[1]{\textbf{class} {#1}:}{\phantom}
\algtext*{EndClass}
\algsetblockdefx[Interface]{Interface}{EndInterface}{}{}{\textbf{Interface:}}{\phantom}
\algtext*{EndInterface}
\algsetblockdefx[Requests]{Requests}{EndRequests}{}{}{\textbf{Requests:}}{\phantom}
\algtext*{EndRequests}
\algsetblockdefx[Logging]{Logging}{EndLogging}{}{}{\textbf{Logging:}}{\phantom}
\algtext*{EndLogging}
\algsetblockdefx[Algorithm]{Algorithm}{EndAlgorithm}{}{}{\textbf{Algorithm:}}{\phantom}
\algtext*{EndAlgorithm}
\algsetblockdefx[Constants]{Constants}{EndConstants}{}{}{\textbf{Constants:}}{\phantom}
\algtext*{EndConstants}
\algsetblockdefx[Upon]{Upon}{EndUpon}{}{}[1]{\textbf{upon} $#1$}{\textbf{end upon}}
\algsetblockdefx[UponReceipt]{UponReceipt}{EndUponReceipt}{}{}[2]{\textbf{upon receipt of} $#1$ \textbf{ from} $#2$}{\textbf{end upon receipt}}
\algsetblockdefx[Structure]{Structure}{EndStructure}{}{}[1]{$#1 = \{$}{$\}$}

\algsetblockdefx[Procedure]{Procedure}{EndProcedure}{}{}[2]{\textbf{procedure} {#1}($#2$)}{\phantom}
\algtext*{EndProcedure}
\algsetblockdefx[Function]{Function}{EndFunction}{}{}[2]{\textbf{function} {#1}($#2$)}{\phantom}
\algtext*{EndFunction}

\algsetblockdefx[Leader]{Leader}{EndLeader}{}{}{\textbf{as} a leader}{\phantom}
\algtext*{EndLeader}
\algsetblockdefx[Replica]{Replica}{EndReplica}{}{}{\textbf{as} a replica}{\phantom}
\algtext*{EndReplica}

\algnewcommand{\BlueComment}[1]{\textcolor{blue}{\hfill\(\triangleright\) #1}}
\algnewcommand{\LineComment}[1]{\textcolor{blue}{ \(\triangleright\) #1}}

\usepackage{bold-extra} 

\usepackage{listings}
\crefname{lstlisting}{listing}{listings}
\Crefname{lstlisting}{Listing}{Listings}

\crefname{code}{line}{lines}
\Crefname{code}{Line}{Lines}

\usepackage{xcolor}

\definecolor{mygreen}{rgb}{0.254,0.572,0.294}
\definecolor{mygray}{rgb}{0.5,0.5,0.5}
\definecolor{myorange}{rgb}{1,0.35,0}
\definecolor{mymauve}{rgb}{0.58,0,0.82}
\definecolor{myblue}{rgb}{0.2,0.4,0.6}

\definecolor{rakos4orange}{RGB}{255,165,0}
\definecolor{rakos4blue}{RGB}{14,48,173}
\definecolor{rakos4lblue}{RGB}{92,172,238}
\definecolor{rakos4dgray}{RGB}{77,77,77}
\definecolor{plainred}{RGB}{211,63,63}
\definecolor{plainorange}{RGB}{221,105,41}

\lstdefinelanguage{Golang}%
  {morekeywords=[1]{package,import,struct,defer,panic,%
     recover,select,var,const,iota, class},%
   morekeywords=[2]{string,uint,uint8,uint16,uint32,uint64,int,int8,int16,%
     int32,int64,bool,float32,float64,complex64,complex128,byte,rune,uintptr,%
     error,interface,node},%
   morekeywords=[3]{map,slice,make,new,nil,len,cap,copy,close,%
     delete,append,real,imag,complex,chan,},%
   morekeywords=[4]{break,continue,goto,switch,case,fallthrough,%
    default,},%
   morekeywords=[5]{Println,Printf,Error,Send},%
   sensitive=true,%
   morecomment=[l]{//},%
   morecomment=[s]{/*}{*/},%
   morestring=[b]",%
   morestring=[s]{`}{`},%
   }

\usepackage{algorithmicx}
\usepackage{algorithm} 
\usepackage{algpseudocode}
\usepackage{xifthen}
\usepackage{graphicx}
\usepackage{wrapfig}

\usepackage{stmaryrd}
\usepackage{amsmath}
\usepackage{amsthm}

\usepackage{cancel}
\usepackage{soul}

\newif\ifcomments
\commentstrue

\author{Pierre Civit}{Sorbonne University, France}{}{}{}
\author{Seth Gilbert}{NUS Singapore, Singapore}{}{}{Supported in part by the Singapore MOE Tier 2 grant MOE-T2EP20122-0014.}
\author{Rachid Guerraoui}{Ecole Polytechnique Fédérale de Lausanne (EPFL), Switzerland}{}{}{}
\author{Jovan Komatovic}{Ecole Polytechnique Fédérale de Lausanne (EPFL), Switzerland}{}{}{}
\author{Matteo Monti}{Ecole Polytechnique Fédérale de Lausanne (EPFL), Switzerland}{}{}{}
\author{Manuel Vidigueira}{Ecole Polytechnique Fédérale de Lausanne (EPFL), Switzerland}{}{}{Supported in part by the FNS (\#200021\_215383).}

\nolinenumbers 

\hideLIPIcs  

\begin{document}


\title{Every Bit Counts in Consensus \\ (Extended Version)}



\authorrunning{P. Civit, S. Gilbert, R. Guerraoui, J. Komatovic, M. Monti, M. Vidigueira} 

\Copyright{Anon} 



\bibliographystyle{plainurl}

\maketitle
  

\begin{abstract}
Consensus enables $n$ processes to agree on a common valid $L$-bit value, despite $t < n/3$ processes being faulty and acting arbitrarily.
A long line of work has been dedicated to improving the worst-case communication complexity of consensus in partial synchrony.
This has recently culminated in the worst-case \emph{word} complexity of $O(n^2)$.
However, the worst-case \emph{bit} complexity of the best solution is still $O(n^2L + n^2\kappa)$ (where $\kappa$ is the security parameter), far from the $\Omega(nL + n^2)$ lower bound.
The gap is significant given the practical use of consensus primitives, where values typically consist of batches of large size ($L > n$).

This paper shows how to narrow the aforementioned gap.
Namely, we present a new algorithm, \name (Disperse, Agree, REtrieve),
that improves upon the $O(n^2L)$ term via a novel dispersal primitive.
DARE achieves $O(n^{1.5}L + n^{2.5}\kappa)$ bit complexity, an effective $\sqrt{n}$-factor improvement over the state-of-the-art (when $L > n\kappa$).
Moreover, we show that employing heavier cryptographic primitives, namely STARK proofs, allows us to devise \darestark, a version of \name which achieves the near-optimal bit complexity of $O(nL + n^2\mathit{poly}(\kappa))$.
Both \name and \darestark achieve optimal $O(n)$ worst-case latency.




\end{abstract}

\maketitle

\section{Introduction} \label{section:introduction}




Byzantine consensus~\cite{Lamport1982} is a fundamental primitive 
in distributed computing.
It has recently risen to prominence due to its use in blockchains~\cite{luu2015scp,buchman2016tendermint,abraham2016solida,chen2016algorand,abraham2016solidus,CGL18,correia2019byzantine} and various forms of state machine replication (SMR)~\cite{adya2002farsite,CL02,kotla2004high,abd2005fault,amir2006scaling,kotla2007zyzzyva,veronese2011efficient,malkhi2019flexible,momose2021multi}. At the same time, the performance of these applications has become directly tied to the performance of consensus and its efficient use of network resources. Specifically, the key limitation on blockchain transaction rates today is network throughput \cite{narwhal, bullshark, chopchop}. This has sparked a large demand for research into Byzantine consensus algorithms with better communication complexity guarantees.

Consensus operates among $n$ processes: each process proposes its value, and all processes eventually agree on a common valid $L$-bit decision.
A process can either be correct or faulty: correct processes follow the prescribed protocol, while faulty processes (up to $t < n / 3$) can behave arbitrarily.
Consensus satisfies the following properties:
\begin{compactitem} 
    \item \emph{Agreement:} No two correct processes decide different values.
    
    \item \emph{Termination:} All correct processes eventually decide.

    \item \emph{(External) Validity:} If a correct process decides a value $v$, then $\mathsf{valid}(v) = \mathit{true}$.
\end{compactitem}
Here, $\mathsf{valid}(\cdot)$ is any predefined logical predicate that indicates whether or not a value is valid.\footnote{For traditional notions of validity, admissible values depend on the proposals of correct processes, e.g., if all correct processes start with value $v$, then $v$ is the only admissible decision. In this paper, we focus on external validity~\cite{Cachin2001}, with the observation that any other validity condition can be achieved by reduction (as shown in~\cite{validity_paper}).}

This paper focuses on improving the worst-case bit complexity of deterministic Byzantine consensus in standard partial synchrony~\cite{Dwork1988}.
The worst-case lower bound is $\Omega(nL + n^2)$ exchanged bits.
This considers all bits sent by correct processes from the moment the network becomes synchronous, i.e., GST (the number of messages sent by correct processes before GST is unbounded due to asynchrony~\cite{Spiegelman2021}).
The $nL$ term comes from the fact that all $n$ processes must receive the decided value at least once, while the $n^2$ term is implied by the seminal Dolev-Reischuk lower bound~\cite{Dolev1985, Spiegelman2021} on the number of messages.
Recently, a long line of work has culminated in Byzantine consensus algorithms which achieve optimal $O(n^2)$ worst-case \emph{word} complexity, where a word is any constant number of values, signatures or hashes~\cite{CivitDGGGKV22,LewisPye}.
However, to the best of our knowledge, no existing algorithm beats the $O(n^2L + n^2\kappa)$ bound on the worst-case bit complexity, where $\kappa$ denotes the security parameter (e.g., the number of bits per hash or signature).
The $n^2L$ term presents a linear gap with respect to the lower bound.

Does this gap matter?
In practice, yes. 
In many cases, consensus protocols are used to agree on a large batch of inputs~\cite{bitcoin, ethereum, chopchop, narwhal, bullshark}.
For example, a block in a blockchain amalgamates many transactions.
Alternatively, imagine that $n$ parties each propose a value, and the protocol agrees on a set of these values.
(This is often known as vector consensus~\cite{BenOr1993, ben1994asynchronous, duan2023practical, neves2005solving,Doudou1998,correia2006consensus}.)  
Typically, the hope is that by batching values/transactions, we can improve the total throughput of the system.
Unfortunately, with current consensus protocols, larger batches do not necessarily yield better performance when applied directly~\cite{decker2016bitcoin}.
This does not mean that batches are necessarily ineffective.
In fact, a recent line of work has achieved significant practical improvements to consensus throughput by entirely focusing on the efficient dissemination of large batches (i.e., large values), so-called ``mempool'' protocols~\cite{narwhal, bullshark, chopchop}.
While these solutions work only optimistically (they perform well in periods of synchrony and without faults), they show that a holistic focus on \emph{bandwidth} usage is fundamental (i.e., bit complexity, and not just word complexity).



\subsection{Contributions}

We introduce DARE (Disperse, Agree, REtrieve), a new Byzantine consensus algorithm for partial synchrony with worst-case $O(n^{1.5}L + n^{2.5}\kappa)$ bit complexity and optimal worst-case $O(n)$ latency. 
Moreover, by enriching DARE with heavier cryptographic primitives, namely STARK proofs, we close the gap near-optimally using only $O(nL + n^2\mathit{poly}(\kappa))$ bits.
Notice that, if you think of $L$ as a batch of $n$ transactions of size $s$, the average communication cost of agreeing on a single transaction is only $\tilde{O}(ns)$ bits -- the same as a best-effort (unsafe) broadcast~\cite{cachin_introduction} of that transaction!

To the best of our knowledge, DARE is the first partially synchronous algorithm to achieve $o(n^2L)$ bit complexity and $O(n)$ latency.
The main idea behind DARE is to separate the problem of agreeing from the problem of retrieving an agreed-upon value (see \Cref{section:technical_overview} for more details).
Figure~\ref{fig:complexities_summary} places DARE in the context of efficient consensus algorithms.

\begin{figure}[ht]
\centering
\begin{tabular}{ |p{3cm}|p{1cm}|p{2.5cm}|p{3cm}|p{2cm}|  }
 \hline
 Protocol & Model & Cryptography & \multicolumn{2}{|c|}{Complexity} \\
 \hline\hline
 & & & $\mathbb{E}[\text{Bits}]^\dagger$ & $\mathbb{E}[\text{Latency}]^\dagger$ \\
 \hline
 ABC \cite{Cachin2001}$^\ddagger$   & Async & PKI, TS \cite{Libert2016}              & $O(n^2L + n^2\kappa + n^{3})$ & $O(1) $\\
 VABA \cite{abraham2019asymptotically}          & Async & above                & $O(n^{2}L + n^2\kappa)$ & $O(1) $\\
 Dumbo-MVBA \cite{Lu2020}    & Async & above + ECC \cite{blahut1983theory}      & $O(nL + n^2\kappa)$ & $O(1) $\\ [0.5ex]
 \hline\hline
 & & & Bits & Latency \\
 \hline
 PBFT \cite{Castro2002, bessani2014state}          & PSync & PKI                  & $O(n^2L + n^{4}\kappa)$ & $O(n) $ \\
 HotStuff \cite{Yin2019}    & PSync & above + TS           & $O(n^2L + n^{3}\kappa)$ & $O(n) $ \\
 Quad \cite{CivitDGGGKV22,LewisPye} & PSync & above                & $O(n^{2}L + n^2\kappa)$ & $O(n) $ \\
 \textbf{DARE} & PSync & above + ECC         & $O(n^{1.5}L + n^{2.5}\kappa)$ & $O(n) $ \\
 \textbf{\textsc{DARE-Stark}} & PSync & above + STARK & $O(nL + n^2\kappa)$ & $O(n) $ \\
 \hline
\end{tabular}
    \caption{Performance of various consensus algorithms with $L$-bit values and $\kappa$-bit security parameter.\\
    $^\dagger$ For asynchronous algorithms, we show the complexity in expectation instead of the worst-case (which is unbounded for deterministic safety guarantees due to the FLP impossibility result \cite{Fischer1985}).\\
    $^\ddagger$ Threshold Signatures (TS) are used to directly improve the original algorithm.}
\label{fig:complexities_summary}
\end{figure}


\subsection{Technical Overview}
\label{section:technical_overview}

\smallskip
\noindent \textbf{The ``curse'' of GST.}
To understand the problem that DARE solves, we must first understand why existing algorithms suffer from an $O(n^2L)$ term.
``Leader-based'' algorithms (such as the state-of-the-art \cite{Momose2021, CivitDGGGKV22, LewisPye}) solve consensus by organizing processes into a rotating sequence of \emph{views}, each with a different leader.
A view's leader broadcasts its value $v$ and drives other processes to decide it.
If all correct processes are timely and the leader is correct, $v$ is decided.

The main issue is that, if synchrony is only guaranteed \emph{eventually} (partial synchrony \cite{Dwork1988}), a view might fail to reach agreement even if its leader is correct: the leader could just be slow (i.e., not yet synchronous).
The inability to distinguish the two scenarios forces protocols to change views even if the current leader is merely ``suspected'' of being faulty.
Since there can be up to $t$ faulty leaders, there must be at least $t+1$ different views.
However, this comes at the risk of sending unnecessary messages if the suspicion proves false, which is what happens in the worst case.

Suppose that, before GST (i.e., the point in time the system becomes synchronous), the first $t$ leaders are correct, but ``go to sleep'' (slow down) immediately before broadcasting their values, and receive no more messages until $\text{GST} + \delta$ due to asynchrony ($\delta$ is the maximum message delay after GST).
Once GST is reached, all $t$ processes wake up and broadcast their value, for a total of $O(tnL) = O(n^2L)$ exchanged bits; this can happen before they have a chance to receive even a single message!
This attack can be considered a ``curse'' of GST: the \emph{adversarial shift} of correct processes in time creates a (seemingly unavoidable) situation where $\Omega(n^2)$ messages are sent at GST (which in this case include $L$ bits each, for a total of $\Omega(n^2L)$).
Figure~\ref{fig:adversarial_shift} illustrates the attack.

\begin{figure}[ht]
   \centering
   \includegraphics[scale=0.8]{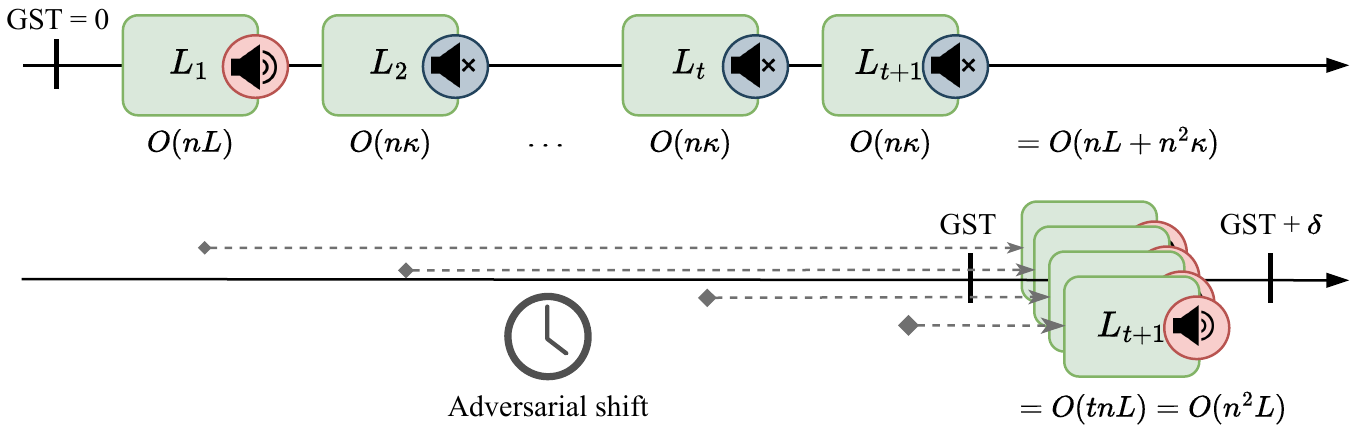}
   \caption{The \emph{adversarial shift} attack on $t+1$ leaders.
   The first line shows how leaders are optimistically ordered in time by the protocol to avoid redundant broadcasts (the blue speaker circle represents an \emph{avoided} redundant broadcast).
   The second line shows how leaders can slow down before GST and overlap at GST, making redundant broadcasts (seem) unavoidable.
   }
   \label{fig:adversarial_shift}
\end{figure}

\smallskip
\noindent \textbf{DARE: Disperse, Agree, REtrieve.}
In a nutshell, DARE follows three phases:
\begin{compactenum}
    \item \textbf{Dispersal}: Processes attempt to disperse their values and obtain a \emph{proof of dispersal} for any value. This proof guarantees that the value is both (1) valid, and (2) retrievable.
    
    \item \textbf{Agreement}: Processes propose a \emph{hash} of the value accompanied by its proof of dispersal to a Byzantine consensus algorithm for small $L$ (e.g., $O(\kappa)$).

    \item \textbf{Retrieval}: Using the decided hash, processes retrieve the corresponding value. The proof of dispersal ensures Retrieval will succeed and output a valid value.
\end{compactenum}
This architecture is inspired by randomized asynchronous Byzantine algorithms~\cite{abraham2019asymptotically,Lu2020} which work with \emph{expected} bit complexity (the worst-case is unbounded in asynchrony~\cite{Fischer1985}).
As these algorithms work in expectation, they can rely on randomness to retrieve a value ($\neq \bot$) that is valid after an expected constant number of tries.
However, in order to achieve the same effect (i.e., a constant number of retrievals) in the worst case in partial synchrony, \name must guarantee that the Retrieval protocol \emph{always} outputs a valid value ($\neq \bot$) \emph{a priori}, which shifts the difficulty of the problem almost entirely to the Dispersal phase.


\smallskip
\noindent \textbf{Dispersal.} To obtain a proof of dispersal, a natural solution is for the leader to broadcast the value $v$.
Correct processes check the validity of $v$ (i.e., if $\mathsf{valid}(v) = \mathit{true}$), store $v$ for the Retrieval protocol, and produce a partial signature attesting to these two facts.
The leader combines the partial signatures into a $(2t+1, n)$-threshold signature (the proof of dispersal), which is sufficient to prove that DARE's Retrieval protocol~\cite{das2021asynchronous} will output a valid value after the Agreement phase.

However, if leaders use best-effort broadcast~\cite{cachin_introduction} (i.e., simultaneously send the value to all other processes), they are still vulnerable to an \emph{adversarial shift} causing $O(n^2L)$ communication.
Instead, we do the following.
First, we use a \emph{view synchronizer}~\cite{Naor2020, Bravo2020, Bravo2022b} to group leaders into \emph{views} in a rotating sequence.
A view has $\sqrt{n}$ leaders and a sequence has $\sqrt{n}$ views.
Leaders of the current view can concurrently broadcast their values while messages of other views are ignored.
Second, instead of broadcasting the value simultaneously to all processes, a leader broadcasts the value to different subgroups of $\sqrt{n}$ processes in intervals of $\delta$ time (i.e., broadcast to the first subgroup, wait $\delta$ time, broadcast to the second subgroup, \ldots) until all processes have received the value.
Neither idea individually is enough to improve over the $O(n^2L)$ term.
However, when they are combined, it becomes possible to balance the communication cost of the synchronizer ($O(n^{2.5}\kappa)$ bits), the maximum cost of an \emph{adversarial shift} attack ($O(n^{1.5}L)$ bits), and the broadcast rate to achieve the improved $O(n^{1.5}L + n^{2.5}\kappa)$ bit complexity with asymptotically optimal $O(\delta n)$ latency as shown in Figure~\ref{fig:dare_overview}.

\begin{figure}[ht]
   \centering
   \includegraphics[scale=0.8]{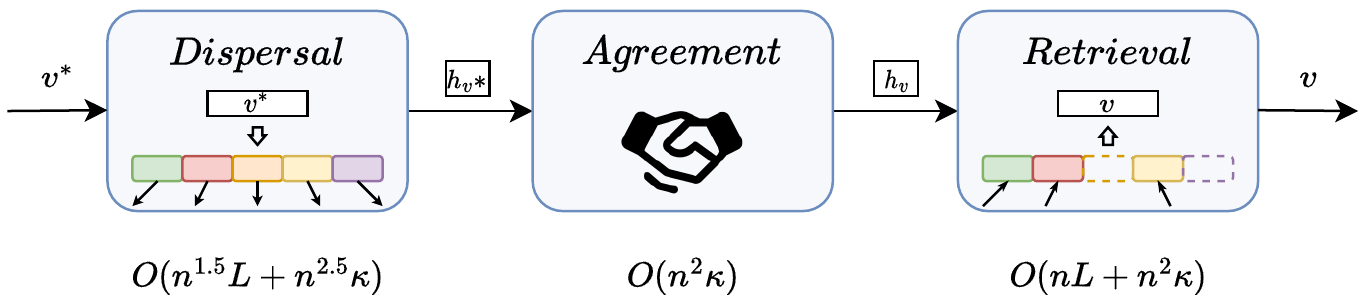}
   \caption{Overview of \name (Disperse, Agree, REtrieve).}
   \label{fig:dare_overview}
\end{figure}

\smallskip
\noindent \textbf{\darestark.}
As we explained, the main cost of the Dispersal phase is associated with obtaining a dispersal proof that a value is valid.
Specifically, it comes from the cost of having to send the entire value ($L$ bits) in a single message.

With Succinct Transparent ARguments of Knowledge (STARKs), we can entirely avoid sending the value in a single message.
STARKs allow a process to compute a proof ($O(\mathit{poly}(\kappa))$ bits) of a statement on some value without having to share that value.
As an example, a process $P_i$ can send $\langle h, \sigma_{\textsf{STARK}} \rangle$ to a process $P_j$, which can use $\sigma_{\textsf{STARK}}$ to verify the statement ``$\exists v: \mathsf{valid}(v) = \mathit{true} \land \mathsf{hash}(v) = h$'', all without $P_j$ ever receiving $v$.
As we detail in \Cref{section:darestark}, by carefully crafting a more complex statement, we can modify DARE's Dispersal and Retrieval phases to function with at most $O(\mathit{poly}(\kappa))$ bit-sized messages, obtaining \darestark.
This yields the overall near-optimal bit complexity of $O(nL + n^2\mathit{poly}(\kappa))$.
Currently, the main drawback of STARKs is their size and computation time in practice\footnote{The associated constants hidden by the ``big O'' notation result in computation in the order of seconds, proofs in the hundreds of KB, and memory usage several times greater \cite{winterfell}.}, which we hope will improve in the future.

\smallskip
\noindent \textbf{Roadmap.}
We discuss related work in \Cref{section:related_work}.
In \Cref{section:preliminaries}, we define the system model.
We give an overview of DARE in \Cref{section:general_framework}.
In \Cref{section:spread}, we detail our Dispersal protocol.
We go over \darestark in \Cref{section:darestark}.
Lastly, we conclude the paper in \Cref{section:conclusion}.
Detailed proofs are relegated to the optional appendix. 




\section{Related Work}\label{section:related_work}

We address the communication complexity of deterministic authenticated Byzantine consensus~\cite{Lamport1982,Cachin2001} in partially synchronous  distributed systems \cite{Dwork1988} for large inputs.
Here, we discuss existing results in closely related contexts, and provide a brief overview of techniques, tools and building blocks which are often employed to tackle Byzantine consensus.\footnote{We use ``consensus'' and ``agreement'' interchangeably.}

\smallskip
\noindent \textbf{Asynchrony}.
In the asynchronous setting, Byzantine agreement is commonly known as Multi-valued Validated Byzantine Agreement, or MVBA~\cite{Cachin2001}.
Due to the FLP impossibility result \cite{Fischer1985}, deterministic Byzantine agreement is unsolvable in asynchrony (which implies unbounded worst-case complexity).
Hence, asynchronous MVBA solutions focus on expected complexity.
This line of work was revitalized by HoneyBadgerBFT \cite{miller2016}, the first practical fully asynchronous MVBA implementation.
Like most other modern asynchronous MVBA protocols, it leverages randomization via a common coin \cite{MostefaouiMR15}, and it terminates in expected $O(\log n)$ time with an expected bit complexity of $O(n^2L + n^3\kappa\log n)$.
\cite{abraham2019asymptotically} improves this to $O(1)$ expected time and $O(n^2L + n^2\kappa)$ expected bits, which is asymptotically optimal with $L,\kappa \in O(1)$.
Their result is later extended by \cite{Lu2020} to large values, improving the complexity to $O(nL + n^2\kappa)$ expected bits.
This matches the best known lower bound \cite{Spiegelman2021, Abraham2019b, Dolev1985}, assuming $\kappa \in O(1)$.

\smallskip
\noindent \textbf{Extension protocols}
\cite{Nayak2020,Ganesh2016,Ganesh2017,Ganesh2021}.
An extension protocol optimizes for long inputs via a reduction to the same problem with small inputs (considered an oracle).
Using extension protocols, several state-of-the-art results were achieved in the authenticated and unauthenticated models, both in synchronous and fully asynchronous settings for Byzantine consensus, Byzantine broadcast and reliable broadcast \cite{Nayak2020}. 
Applying the extension protocol of \cite{Nayak2020} to \cite{Momose2021}, synchronous Byzantine agreement can be implemented with optimal resiliency ($t<n/2$) and a bit complexity of $O(nL + n^2\kappa)$.
Interestingly, it has been demonstrated that synchronous Byzantine agreement can be implemented with a bit complexity of $O(n (L+\mathit{poly}(\kappa)))$ using randomization \cite{Bhangale2022}. The Dolev-Reischuk bound \cite{Dolev1985} is not violated in this case since the implementation tolerates a negligible (with $\kappa$) probability of failure, whereas the bound holds for deterministic protocols.
In asynchrony, by applying the (asynchronous) extension protocol of \cite{Nayak2020} to \cite{abraham2019asymptotically}, the same asymptotic result as \cite{Lu2020} is achieved, solving asynchronous MVBA with an expected bit complexity of $O(nL + n^2\kappa)$.

Unconditionally secure Byzantine agreement with large inputs has been addressed by \cite{Chen2020, Chen2021} under synchrony and \cite{Li2021} under asynchrony, assuming a common coin (implementable via unconditionally-secure Asynchronous Verifiable Secret Sharing~\cite{Choudhury2023}).
Despite \cite{kaklamanis2022} utilizing erasure codes to alleviate leader bottleneck, and the theoretical construction of \cite{validity_paper} with exponential latency, there is, to the best of our knowledge, no viable extension protocol for Byzantine agreement in partial synchrony achieving results similar to ours ($o(n^2L)$).

\smallskip
\noindent \textbf{Error correction.}
Coding techniques, such as erasure codes \cite{blahut1983theory, Hendricks2007b, alhaddad2021succinct} or error-correction codes \cite{reed1960polynomial, BenOr1993}, appear in state-of-the-art implementations of various distributed tasks: Asynchronous Verifiable Secret Sharing (AVSS) against a computationally bounded \cite{das2021asynchronous, Yurek2022, Shoup2023} or unbounded \cite{Choudhury2023} adversary,
Random Beacon \cite{Das2022}, Atomic Broadcast in both the asynchronous \cite{Gagol2019, Keidar2021} and partially synchronous \cite{Camenisch2022} settings, Information-Theoretic (IT) Asynchronous State Machine Replication (SMR) \cite{Duan2022}, Gradecast in synchrony and Reliable Broadcast in asynchrony \cite{Abraham2022}, Asynchronous Distributed Key Generation (ADKG) \cite{das2021asynchronous, Das2022b}, Asynchronous Verifiable Information Dispersal (AVID) \cite{Alhaddad2022}, Byzantine Storage \cite{Dobre2019, Androulaki2014, Hendricks2007}, and MVBA \cite{Lu2020, Nayak2020}.
Coding techniques are often used to reduce the worst-case complexity by allowing a group of processes to balance and share the cost of sending a value to an individual (potentially faulty) node and are also used in combination with other techniques, such as commitment schemes \cite{catalano2013vector, Lu2020}.

\smallskip
\noindent {We now list several problems related to or used in solving Byzantine agreement.}

\smallskip
\noindent \textbf{Asynchronous Common Subset (ACS)}.
The goal in ACS \cite{BenOr1993, ben1994asynchronous, duan2023practical} (also known as Vector Consensus \cite{neves2005solving,Doudou1998,correia2006consensus}) is to agree on a subset of $n - t$ proposals. 
When considering a generalization of the validity property, this problem represents the strongest variant of consensus \cite{validity_paper}.
Atomic Broadcast can be trivially reduced to ACS \cite{Doudou1998, Cachin2001, miller2016}. 
There are well-known simple asynchronous constructions that allow for the reduction of ACS to either (1) Reliable Broadcast and Binary Byzantine Agreement \cite{ben1994asynchronous}, or (2) MVBA \cite{Cachin2001} in the authenticated setting, where the validation predicate requires the output to be a vector of signed inputs from at least $n-t$ parties.
The first reduction enables the implementation of ACS with a cubic bit complexity, using the broadcast of \cite{Abraham2022}. 
The second reduction could be improved further with a more efficient underlying MVBA protocol, such as \darestark.

\smallskip
\noindent \textbf{Asynchronous Verifiable Information Dispersal (AVID)}. 
AVID \cite{cachin2005} is a form of "retrievable" broadcast that allows the dissemination of a value while providing a cryptographic proof that it can be retrieved. 
This primitive can be implemented with a total dispersal cost of $O(L + n^2\kappa)$ bits exchanged and a retrieval cost of $O(L + n\kappa)$ per node, relying only on the existence of collision-resistant hash functions \cite{Alhaddad2022}. 
AVID is similar to our Dispersal and Retrieval phases, but has two key differences.
First, AVID's retrieval protocol only guarantees that a valid value will be retrieved if the original process dispersing the information was correct.
Second, it is a broadcast protocol, having stricter delivery guarantees for each process.
Concretely, if a correct process initiates the AVID protocol, it should eventually disperse its \emph{own} value.
In contrast, we only require that a correct process obtains a proof of dispersal for \emph{some} value.

\smallskip
\noindent \textbf{Provable Broadcast (PB) and Asynchronous Provable Dispersal Broadcast (APDB)}.
PB \cite{Abraham2019a} is a primitive used to acquire a succinct proof of external validity. 
It is similar to our Dispersal phase, including the algorithm itself, but without the provision of a proof of dispersal (i.e., retrievability, only offering proof of validity). 
The total bit complexity for $n$ PB-broadcasts from distinct processes amounts to $O(n^2L)$.
APDB \cite{Lu2020} represents an advancement of AVID, drawing inspiration from PB. 
It sacrifices PB's validity guarantees to incorporate AVID's dissemination and retrieval properties.
By leveraging the need to retrieve and validate a value a constant number of times in expectation, \cite{Lu2020} attains optimal $O(nL + n^2\kappa)$ expected complexity in asynchrony.
However, this approach falls short in the worst-case scenario of a partially synchronous solution, where $n$ reconstructions would cost $\Omega(n^2L)$.
    
\smallskip
\noindent \textbf{Asynchronous Data Dissemination (ADD)}. 
In ADD \cite{das2021asynchronous}, a subset of $t+1$ correct processes initially share a common $L$-sized value $v$, and the goal is to disseminate $v$ to all correct processes, despite the presence of up to $t$ Byzantine processes. The approach of \cite{das2021asynchronous} is information-theoretically secure, tolerates up to one-third malicious nodes and has a bit complexity of $O(nL + n^2\log n)$.
(In \name, we rely on ADD in a ``closed-box'' manner; see \Cref{section:general_framework}.)

\section{Preliminaries} \label{section:preliminaries}



\noindent\textbf{Processes.}
We consider a static set $\mathsf{Process} = \{P_1, P_2, ..., P_n\}$ of $n = 3t + 1$ processes, out of which (at most) $t > 0$ can be Byzantine and deviate arbitrarily from their prescribed protocol.
A Byzantine process is said to be \emph{faulty}; a non-faulty process is said to be \emph{correct}.
Processes communicate by exchanging messages over an authenticated point-to-point network.
Furthermore, the communication network is reliable: if a correct process sends a message to a correct process, the message is eventually received.
Processes have local hardware clocks.
Lastly, we assume that local steps of processes take zero time, as the time needed for local computation is negligible compared to the message delays.

\smallskip
\noindent\textbf{Partial synchrony.}
We consider the standard partially synchronous model~\cite{Dwork1988}.
For every execution, there exists an unknown Global Stabilization Time (GST) and a positive duration $\delta$ such that the message delays are bounded by $\delta$ after GST.
We assume that $\delta$ is known by processes.
All correct processes start executing their prescribed protocol by GST.
The hardware clocks of processes may drift arbitrarily before GST, but do not drift thereafter.
We underline that our algorithms require minimal changes to preserve their correctness even if $\delta$ is unknown (these modifications are specified in \Cref{appendix:dare_unknown}), although their complexity might be higher.

\smallskip
\noindent \textbf{Cryptographic primitives.}
Throughout the paper, $\mathsf{hash}(\cdot)$ denotes a collision-resistant hash function.
The codomain of the aforementioned $\mathsf{hash}(\cdot)$ function is denoted by $\mathsf{Hash\_Value}$.

Moreover, we assume a $(k, n)$-threshold signature scheme~\cite{Shoup00}, where $k = n - t = 2t + 1$.
In this scheme, each process holds a distinct private key, and there is a single public key.
Each process $P_i$ can use its private key to produce a partial signature for a message $m$ by invoking $\mathsf{share\_sign}_i(m)$.
A set of partial signatures $S$ for a message $m$ from $k$ distinct processes can be combined into a single threshold signature for $m$ by invoking $\mathsf{combine}(S)$; a threshold signature for $m$ proves that $k$ processes have (partially) signed $m$.
Furthermore, partial and threshold signatures can be verified: given a message $m$ and a signature $\Sigma_m$, $\mathsf{verify\_sig}(m, \Sigma_m)$ returns $\mathit{true}$ if and only if $\Sigma_m$ is a valid signature for $m$.
Where appropriate, the verifications are left implicit.
We denote by $\mathsf{P\_Signature}$ and $\mathsf{T\_Signature}$ the set of partial and threshold signatures, respectively.
The size of cryptographic objects (i.e., hashes, signatures) is denoted by $\kappa$; we assume that $\kappa > \log n$.\footnote{For $\kappa \leq \log n$, $t \in O(n)$ faulty processes would have computational power exponential in $\kappa$, breaking cryptographic hardness assumptions.}

\smallskip
\noindent \textbf{Reed-Solomon codes~\cite{reed1960polynomial}.}
Our algorithms rely on Reed-Solomon (RS) codes~\cite{reed1960}.
Concretely, \name utilizes (in a ``closed-box'' manner) an algorithm which internally builds upon error-correcting RS codes.
\darestark directly uses RS erasure codes (no error correction is required).

We use $\mathsf{encode}(\cdot)$ and $\mathsf{decode}(\cdot)$ to denote RS' encoding and decoding algorithms.
In a nutshell, $\mathsf{encode}(\cdot)$ takes a value $v$, chunks it into the coefficients of a polynomial of degree $t$ (the maximum number of faults), and outputs $n$ (the total number of processes) evaluations of the polynomial (RS symbols); $\mathsf{Symbol}$ denotes the set of RS symbols.
$\mathsf{decode}(\cdot)$ takes a set of $t + 1$ RS symbols $S$ and interpolates them into a polynomial of degree $t$, whose coefficients are concatenated and output.

\smallskip
\noindent \textbf{Complexity of Byzantine consensus.}
Let $\mathsf{Consensus}$ be a partially synchronous Byzantine consensus algorithm, and let $\mathcal{E}(\mathsf{Consensus})$ denote the set of all possible executions.
Let $\alpha \in \mathcal{E}(\mathsf{Consensus})$ be an execution, and $t_d(\alpha)$ be the first time by which all correct processes have decided in $\alpha$.
The bit complexity of $\alpha$ is the total number of bits sent by correct processes during the time period $[\text{GST}, \infty)$.
The latency of $\alpha$ is $\max(0, t_d(\alpha) - \text{GST})$.

The \emph{bit complexity} of $\mathsf{Consensus}$ is defined as
\begin{equation*}
\max_{\alpha \in \mathcal{E}(\mathsf{Consensus})}\bigg\{\text{bit complexity of } \alpha\bigg\}.
\end{equation*}

Similarly, the \emph{latency} of $\mathsf{Consensus}$ is defined as 
\begin{equation*}
\max_{\alpha \in \mathcal{E}(\mathsf{Consensus})}\bigg\{\text{latency of } \alpha\bigg\}.
\end{equation*}

\section{\name} \label{section:general_framework}

This section presents \name (Disperse, Agree, REtrieve), which is composed of three algorithms: 
\begin{compactenum}
    \item \spread, which disperses the proposals; 

    \item \q, which ensures agreement on the hash of a previously dispersed proposal; and

    \item \rec, which rebuilds the proposal corresponding to the agreed-upon hash.
\end{compactenum}

We start by introducing the aforementioned building blocks (\Cref{subsection:building_blocks}).
Then, we show how they are composed into \name (\Cref{subsection:algorithm}).
Finally, we prove the correctness and complexity of \name (\Cref{subsection:dare_proof}).



\subsection{Building Blocks: Overview} \label{subsection:building_blocks}

In this subsection, we formally define the three building blocks of \name.
Concretely, we define their interface and properties, as well as their complexity.

\subsubsection{\spread} \label{subsubsection:spread}

\noindent \textbf{Interface \& properties.}
\spread solves a problem similar to that of AVID~\cite{cachin2005}.
In a nutshell, each correct process aims to disperse its value to all correct processes: eventually, all correct processes acquire a proof that a value with a certain hash has been successfully dispersed.

Concretely, \spread exposes the following interface:
\begin{compactitem}
    \item \textbf{request} $\mathsf{disperse}(v \in \mathsf{Value})$: a process disperses a value $v$; each correct process invokes $\mathsf{disperse}(v)$ exactly once and only if $\mathsf{valid}(v) = \mathit{true}$.

    \item \textbf{indication} $\mathsf{acquire}(h \in \mathsf{Hash\_Value}, \Sigma_h \in \mathsf{T\_Signature})$: a process acquires a pair $(h, \Sigma_h)$.
\end{compactitem}
We say that a correct process \emph{obtains} a threshold signature (resp., a value) if and only if it stores the signature (resp., the value) in its local memory.
(Obtained values can later be retrieved by all correct processes using \rec; see \Cref{subsubsection:rec_overview} and \Cref{algorithm:root}.) 
\spread ensures the following:
\begin{compactitem}
    \item \emph{Integrity:} If a correct process acquires a hash-signature pair $(h, \Sigma_h)$, then $\mathsf{verify\_sig}(h, \Sigma_h) = \mathit{true}$.

    \item \emph{Termination:} Every correct process eventually acquires at least one hash-signature pair.

    \item \emph{Redundancy:} Let a correct process obtain a threshold signature $\Sigma_h$ such that $\mathsf{verify\_sig}(h, \Sigma_h) = \mathit{true}$, for some hash value $h$.
    Then, (at least) $t + 1$ correct processes have obtained a value $v$ such that (1) $\mathsf{hash}(v) = h$, and (2) $\mathsf{valid}(v) = \mathit{true}$.
    
    
\end{compactitem}
Note that it is not required for all correct processes to acquire the same hash value (nor the same threshold signature).
Moreover, the specification allows for multiple acquired pairs.

\smallskip
\noindent \textbf{Complexity.}
\spread exchanges $O(n^{1.5}L + n^{2.5}\kappa)$ bits after GST. 
Moreover, it terminates in $O(n)$ time after GST.

\smallskip
\noindent \textbf{Implementation.} The details on \spread's implementation are relegated to \Cref{section:spread}.






\subsubsection{\q} \label{subsubsection:quad_overview}

\noindent \textbf{Interface \& properties.}
\q is a Byzantine consensus algorithm.\footnote{Recall that the interface and properties of Byzantine consensus algorithms are introduced in \Cref{section:introduction}.}
In \q, processes propose and decide pairs $(h \in \mathsf{Hash\_Value}, \Sigma_h \in \mathsf{T\_Signature})$; moreover, $\mathsf{valid}(h, \Sigma_h) \equiv \mathsf{verify\_sig}(h, \Sigma_h)$.

\smallskip
\noindent \textbf{Complexity.} \q achieves $O(n^2\kappa)$ bit complexity and $O(n)$ latency.

\smallskip
\noindent \textbf{Implementation.}
We ``borrow'' the implementation from~\cite{CivitDGGGKV22}.
In brief, \q is a ``leader-based'' consensus algorithm whose computation unfolds in views.
Each view has a single leader, and it employs a ``leader-to-all, all-to-leader'' communication pattern.
\q's safety relies on standard techniques~\cite{Yin2019,Castro2002,Buchman2018,LewisPye}: (1) quorum intersection (safety within a view), and (2) ``locking'' mechanism (safety across multiple views).
As for liveness, \q guarantees termination once all correct processes are in the same view (for ``long enough'' time) with a correct leader.
(For full details on \q, see~\cite{CivitDGGGKV22}.)

\subsubsection{\rec} \label{subsubsection:rec_overview}

\noindent \textbf{Interface \& properties.}
In \rec, each correct process starts with either (1) some value, or (2) $\bot$.
Eventually, all correct processes output the same value.
Formally, \rec exposes the following interface:
\begin{compactitem}
    \item \textbf{request} $\mathsf{input}(v \in \mathsf{Value} \cup \{\bot\})$: a process inputs a value or $\bot$; each correct process invokes $\mathsf{input}(\cdot)$ exactly once.
    Moreover, the following is assumed:
    \begin{compactitem}
        \item No two correct processes invoke $\mathsf{input}(v_1 \in \mathsf{Value})$ and $\mathsf{input}(v_2 \in \mathsf{Value})$ with $v_1 \neq v_2$.
    
        \item At least $t + 1$ correct processes invoke $\mathsf{input}(v \in \mathsf{Value})$ (i.e., $v \neq \bot)$.
    \end{compactitem}

    \item \textbf{indication} $\mathsf{output}(v' \in \mathsf{Value})$: a process outputs a value $v'$.
\end{compactitem}
The following properties are ensured:
\begin{compactitem}
    \item \emph{Agreement:} No two correct processes output different values.

    \item \emph{Validity:} Let a correct process input a value $v$.
    No correct process outputs a value $v' \neq v$.


    \item \emph{Termination:} Every correct process eventually outputs a value.
\end{compactitem}

\smallskip
\noindent \textbf{Complexity.} \rec exchanges $O(nL + n^2\log n)$ bits after GST (and before every correct process outputs a value).
Moreover, \rec terminates in $O(1)$ time after GST.

\smallskip
\noindent \textbf{Implementation.}
\rec's implementation is ``borrowed'' from~\cite{das2021asynchronous}.
In summary, \rec relies on Reed-Solomon codes~\cite{reed1960polynomial} to encode the input value $v \neq \bot$ into $n$ symbols.
Each correct process $Q$ which inputs $v \neq \bot$ to \rec encodes $v$ into $n$ RS symbols $s_1, s_2, ..., s_n$.
$Q$ sends each RS symbol $s_i$ to the process $P_i$.
When $P_i$ receives $t + 1$ identical RS symbols $s_i$, $P_i$ is sure that $s_i$ is a ``correct'' symbol (i.e., it can be used to rebuild $v$) as it was computed by at least one correct process.
At this moment, $P_i$ broadcasts $s_i$.
Once each correct process $P$ receives $2t + 1$ (or more) RS symbols, $P$ tries to rebuild $v$ (with some error-correction).
(For full details on \rec, see~\cite{das2021asynchronous}.)


\subsection{Pseudocode} \label{subsection:algorithm}

\Cref{algorithm:root} gives \name's pseudocode.
We explain it from the perspective of a correct process $P_i$.
An execution of \name consists of three phases (each of which corresponds to one building block):
\begin{compactenum}

    \item \emph{Dispersal:}
    Process $P_i$ disperses its proposal $v_i$ using \spread (line~\ref{line:disseminate_value}).
    Eventually, $P_i$ acquires a hash-signature pair $(h_i, \Sigma_{i})$ (line~\ref{line:acquire}) due to the termination property of \spread.

    \item \emph{Agreement:}
    Process $P_i$ proposes the previously acquired hash-signature pair $(h_i, \Sigma_i)$ to \q (line~\ref{line:quad_propose}).
    As \q satisfies termination and agreement, all correct processes eventually agree on a hash-signature pair $(h, \Sigma_h)$ (line~\ref{line:quad_decide}).

    \item \emph{Retrieval:}
    Once process $P_i$ decides $(h, \Sigma_h)$ from \q, it checks whether it has previously obtained a value $v$ with $\mathsf{hash}(v) = h$ (line~\ref{line:check_obtained}).
    If it has, $P_i$ inputs $v$ to \rec; otherwise, $P_i$ inputs $\bot$ (line~\ref{line:vec_rec_input}).
    The required preconditions for \rec are met:
    \begin{compactitem}
        \item No two correct processes input different non-$\bot$ values to \rec as $\mathsf{hash}(\cdot)$ is collision-resistant.

        \item At least $(t + 1)$ correct processes input a value (and not $\bot$) to \rec.
        Indeed, as $\Sigma_h$ is obtained by a correct process, $t + 1$ correct processes have obtained a value $v \neq \bot$ with $\mathsf{hash}(v) = h$ (due to redundancy of \spread), and all of these processes input $v$.
    \end{compactitem}
    Therefore, all correct processes (including $P_i$) eventually output the same value $v'$ from \rec (due to the termination property of \rec; line~\ref{line:vec_rec_output}), which represents the decision of \name (line~\ref{line:decide}).
    Note that $v' = v \neq \bot$ due to the validity of \rec.
\end{compactenum}

\begin{algorithm} [t]
\caption{\name: Pseudocode (for process $P_i$)}
\label{algorithm:root}
\begin{algorithmic} [1]
\footnotesize

\State \textbf{Uses:}

\smallskip
\State \hskip2em \LineComment{bits: $O(n^{1.5}L + n^{2.5}\kappa)$, latency: $O(n)$ (see \Cref{section:spread})}
\State \hskip2em \spread, \textbf{instance} $\mathit{disperser}$

\smallskip
\State \hskip2em \LineComment{bits: $O(n^2\kappa)$, latency: $O(n)$ (see~\cite{CivitDGGGKV22})}
\State \hskip2em \q, \textbf{instance} $\mathit{agreement}$

\smallskip
\State \hskip2em \LineComment{bits: $O(nL + n^2 \log n)$, latency: $O(1)$ (see~\cite{das2021asynchronous})}
\State \hskip2em \rec, \textbf{instance} $\mathit{retriever}$



\smallskip
\State \textbf{upon} $\mathsf{propose}(v_i \in \mathsf{Value})$:
\State \hskip2em \textbf{invoke} $\mathit{disperser}.\mathsf{disperse}(v_i)$ \label{line:disseminate_value}


\smallskip
\State \textbf{upon} $\mathit{disperser}.\mathsf{acquire}(h_i \in \mathsf{Hash\_Value}, \Sigma_i \in \mathsf{T\_Signature})$: \label{line:acquire}
\State \hskip2em \textbf{invoke} $\mathit{agreement}.\mathsf{propose}(h_i, \Sigma_i)$ \label{line:quad_propose}

\smallskip
\State \textbf{upon} $\mathit{agreement}.\mathsf{decide}(h \in \mathsf{Hash\_Value}, \Sigma_h \in \mathsf{T\_Signature})$: \label{line:quad_decide} \BlueComment{$P_i$ obtains $\Sigma_h$}
\State \hskip2em $v \gets $ an obtained value such that $\mathsf{hash}(v) = h$ (if such a value was not obtained, $v = \bot$) \label{line:check_obtained}
\State \hskip2em \textbf{invoke} $\mathit{retriever}.\mathsf{input}(v)$ \label{line:vec_rec_input}

\smallskip
\State \textbf{upon} $\mathit{retriever}.\mathsf{output}(\mathsf{Value} \text{ } v')$: \label{line:vec_rec_output}
\State \hskip2em \textbf{trigger} $\mathsf{decide}(v')$ \label{line:decide}
\end{algorithmic}
\end{algorithm}

\subsection{Proof of Correctness \& Complexity} \label{subsection:dare_proof}

We start by proving the correctness of \name.

\begin{theorem} \label{lemma:root_helper}
\name is correct.
\end{theorem}
\begin{proof}
Every correct process starts the dispersal of its proposal (line~\ref{line:disseminate_value}).
Due to the termination property of \spread, every correct process eventually acquires a hash-signature pair (line~\ref{line:acquire}).
Hence, every correct process eventually proposes to \q (line~\ref{line:quad_propose}), which implies that every correct process eventually decides the same hash-signature pair $(h, \Sigma_h)$ from \q (line~\ref{line:quad_decide}) due to the agreement and termination properties of \q.

As $(h, \Sigma_h)$ is decided by all correct processes, at least $t + 1$ correct processes $P_i$ have obtained a value $v$ such that (1) $\mathsf{hash}(v) = h$, and (2) $\mathsf{valid}(v) = \mathit{true}$ (due to the redundancy property of \spread).
Therefore, all of these correct processes input $v$ to \rec (line~\ref{line:vec_rec_input}).
Moreover, no correct process inputs a different value (as $\mathsf{hash}(\cdot)$ is collision-resistant).
Thus, the conditions required by \rec are met, which implies that all correct processes eventually output the same valid value (namely, $v$) from \rec (line~\ref{line:vec_rec_output}), and decide it (line~\ref{line:decide}). 
\end{proof}



Next, we prove the complexity of \name.

\begin{theorem}
\name achieves $O(n^{1.5}L + n^{2.5}\kappa)$ bit complexity and $O(n)$ latency.   
\end{theorem}
\begin{proof}
As \name is a sequential composition of its building blocks, its complexity is the sum of the complexities of (1) \spread, (2) \q, and (3) \rec.
Hence, the bit complexity is
\begin{equation*}
    \underbrace{O(n^{1.5}L + n^{2.5}\kappa)}_\text{\spread} + \underbrace{O(n^2\kappa)}_\text{\q} + \underbrace{O(nL + n^2\log n)}_\text{\rec} = O(n^{1.5}L + n^{2.5}\kappa).
\end{equation*}
Similarly, the latency is $O(n)$.
\end{proof}
\section{\spread: Implementation \& Analysis} \label{section:spread}

This section focuses on \spread.
Namely, we present its implementation (\Cref{subsection:spread_algorithm}), and (informally) analyze its correctness and complexity (\Cref{subsection:informal_analysis}).
Formal proofs are can be found in \Cref{appendix:spread_proof}.

\subsection{Implementation} \label{subsection:spread_algorithm}

\spread's pseudocode is given in \Cref{algorithm:spread}.
In essence, each execution unfolds in \emph{views}, where each view has $X$ \emph{leaders} ($0 < X \leq n$ is a generic parameter); the set of all views is denoted by $\mathsf{View}$.
Given a view $V$, $\mathsf{leaders}(V)$ denotes the $X$-sized set of leaders of the view $V$.
In each view, a leader disperses its value to $Y$-sized groups of processes ($0 < Y \leq n$ is a generic parameter) at a time (line~\ref{line:multicast_dispersal}), with a $\delta$-waiting step in between (line~\ref{line:wait_spread}).
Before we thoroughly explain the pseudocode, we introduce \sync, \spread's view synchronization~\cite{CivitDGGGKV22,Yin2019,LewisPye} algorithm.

\smallskip
\noindent \textbf{\sync.}
Its responsibility is to bring all correct processes to the same view with a correct leader for (at least) $\Delta = \delta\frac{n}{Y} + 3\delta$ time.
Precisely, \sync exposes the following interface: 
\begin{compactitem}
    \item \textbf{indication} $\mathsf{advance}(V \in \mathsf{View})$: a process enters a new view $V$.
\end{compactitem}
\sync guarantees \emph{eventual synchronization}: there exists a time $\tau_{\mathit{sync}} \geq \text{GST}$ (\emph{synchronization time}) such that (1) all correct processes are in the same view $V_{\mathit{sync}}$ (\emph{synchronization view}) from time $\tau_{\mathit{sync}}$ to (at least) time $\tau_{\mathit{sync}} + \Delta$, and (2) $V_{\mathit{sync}}$ has a correct leader.
We denote by $V_{\mathit{sync}}^*$ the smallest synchronization view, whereas $\tau_{\mathit{sync}}^*$ denotes the first synchronization time.
Similarly, $V_{\mathit{max}}$ denotes the greatest view entered by a correct process before GST.\footnote{When such a view does not exist, $V_{\mathit{max}}=0$}

The implementation of \sync (see \Cref{subsection:sync_pseudocode}) is highly inspired by \textsc{RareSync}, a view synchronization algorithm introduced in~\cite{CivitDGGGKV22}.
In essence, when a process enters a new view, it stays in the view for $O(\Delta) = O(\frac{n}{Y})$ time.
Once it wishes to proceed to the next view, the process engages in an ``all-to-all'' communication step (which exchanges $O(n^2\kappa)$ bits); this step signals the end of the current view, and the beginning of the next one.
Throughout views, leaders are rotated in a round-robin manner: each process is a leader for exactly one view in any sequence of $\frac{n}{X}$ consecutive views.
As $O(\frac{n}{X})$ views (after GST) are required to reach a correct leader, \sync exchanges $O(\frac{n}{X}) \cdot O(n^2\kappa) = O(\frac{n^3\kappa}{X})$ bits (before synchronization, i.e., before $\tau_{\mathit{sync}}^* + \Delta$); since each view takes $O(\frac{n}{Y})$ time, synchronization is ensured within $O(\frac{n}{X}) \cdot O(\frac{n}{Y}) = O(\frac{n^2}{XY})$ time.

\spread relies on the following properties of \sync (along with eventual synchronization):
\begin{compactitem}
    \item \emph{Monotonicity:} Any correct process enters monotonically increasing views.

    \item \emph{Stabilization:} Any correct process enters a view $V \geq V_{\mathit{max}}$ by time $\text{GST} + 3\delta$.

    \item \emph{Limited entrance:} In the time period $[\text{GST}, \text{GST} + 3\delta)$, any correct process enters $O(1)$ views.

    \item \emph{Overlapping:} For any view $V > V_{\mathit{max}}$, all correct processes overlap in $V$ for (at least) $\Delta$ time.
    

    \item \emph{Limited synchronization view:} $V_{\mathit{sync}}^* - V_{\mathit{max}} = O(\frac{n}{X})$.

    \item \emph{Complexity:} \sync exchanges $O(\frac{n^3\kappa}{X})$ bits during the time period $[\text{GST}, \tau_{\mathit{sync}}^* + \Delta]$, and it synchronizes all correct processes within $O(\frac{n^2}{XY})$ time after GST ($\tau_{\mathit{sync}}^* + \Delta - \text{GST} = O(\frac{n^2}{XY})$).
\end{compactitem}
The aforementioned properties of \sync are formally proven in \Cref{subsection:sync_pseudocode}.

\smallskip
\noindent \textbf{Algorithm description.}
Correct processes transit through views based on \sync's indications (line~\ref{line:sync_indication}): when a correct process receives $\mathsf{advance}(V)$ from \sync, it stops participating in the previous view and starts participating in $V$.

Once a correct leader $P_l$ enters a view $V$, it disperses its proposal via \textsc{dispersal} messages.
As already mentioned, $P_l$ sends its proposal to $Y$-sized groups of processes (line~\ref{line:multicast_dispersal}) with a $\delta$-waiting step in between (line~\ref{line:wait_spread}).
When a correct (non-leader) process $P_i$ (which participates in the view $V$) receives a \textsc{dispersal} message from $P_l$, $P_i$ checks whether the dispersed value is valid (line~\ref{line:receive_dispersal}).
If it is, $P_i$ partially signs the hash of the value, and sends it back to $P_l$ (line~\ref{line:send_ack}).
When $P_l$ collects $2t + 1$ \textsc{ack} messages, it (1) creates a threshold signature for the hash of its proposal (line~\ref{line:combine_ack}), and (2) broadcasts the signature (along with the hash of its proposal) to all processes via a \textsc{confirm} message (line~\ref{line:broadcast_confirm}).
Finally, when $P_l$ (or any other correct process) receives a \textsc{confirm} message (line~\ref{line:receive_confirm}), it (1) acquires the received hash-signature pair (line~\ref{line:acquire_spread}), (2) disseminates the pair to ``help'' the other processes (line~\ref{line:echo_confirm}), and (3) stops executing \spread (line~\ref{line:stop_executing_spread}).

\begin{algorithm} [h]
\caption{\spread: Pseudocode (for process $P_i$)}
\label{algorithm:spread}
\begin{algorithmic} [1]
\footnotesize
\State \textbf{Uses:}
\State \hskip2em \sync, \textbf{instance} $\mathit{sync}$ \label{line:view_sync} \BlueComment{ensures a $\Delta = \delta\frac{n}{Y} + 3\delta$ overlap in a view with a correct leader}

\smallskip
\State \textbf{upon} $\mathsf{init}$:
\State \hskip2em $\mathsf{Value}$ $\mathit{proposal}_i \gets \bot$
\State \hskip2em $\mathsf{Integer}$ $\mathit{received\_acks}_i \gets 0$
\State \hskip2em $\mathsf{Map}(\mathsf{Hash\_Value} \to \mathsf{Value})$ $\mathit{obtained\_values}_i \gets \text{empty}$

\smallskip
\State \textbf{upon} $\mathsf{disperse}(\mathsf{Value} \text{ } v_i)$:
\State \hskip2em $\mathit{proposal}_i \gets v_i$
\State \hskip2em \textbf{start} $\mathit{sync}$

\smallskip
\State \textbf{upon} $\mathit{sync}.\mathsf{advance}(\mathsf{View} \text{ } V)$: \BlueComment{$P_i$ stops participating in the previous view} \label{line:sync_indication}
\State \hskip2em \textcolor{blue}{\(\triangleright\) First part of the view}
\State \hskip2em \textbf{if} $P_i \in \mathsf{leaders}(V)$:
\State \hskip4em \textbf{for} $\mathsf{Integer} \text{ } k \gets 1 \text{ to } \frac{n}{Y}$: \label{line:loop_dispersal}
\State \hskip6em \textbf{send} $\langle \textsc{dispersal}, \mathit{proposal}_i \rangle$ to $P_{(k - 1)Y + 1}, P_{(k - 1)Y + 2}, ..., P_{kY}$ \label{line:multicast_dispersal}
\State \hskip6em \textbf{wait} $\delta$ time \label{line:wait_spread}

\smallskip
\State \hskip2em every process:
\State \hskip4em \textbf{upon} reception of $\langle \textsc{dispersal}, \mathsf{Value} \text{ } v_j \rangle$ from process $P_j \in \mathsf{leaders}(V)$ and $\mathsf{valid}(v_j) = \mathit{true}$: \label{line:receive_dispersal}
\State \hskip6em $\mathsf{Hash\_Value}$ $h \gets \mathsf{hash}(v_j)$
\State \hskip6em $\mathit{obtained\_values}_i[h] \gets v_j$ \label{line:obtain}
\label{line:obtained_value}
\State \hskip6em \textbf{send} $\langle \textsc{ack}, \mathsf{share\_sign}_i(h) \rangle$ to $P_j$ \label{line:send_ack} \label{line:share_sign_hash}
 
\medskip
\State \hskip2em \textcolor{blue}{\(\triangleright\) Second part of the view}
\State \hskip2em \textbf{if} $P_i \in \mathsf{leaders}(V)$:
\State \hskip4em \textbf{upon} exists $\mathsf{Hash\_Value}$ $h$ such that $\langle \textsc{ack}, h, \cdot \rangle$ has been received from $2t + 1$ processes: \label{line:receive_ack}
\State \hskip6em $\mathsf{T\_Signature}$ $\Sigma_h \gets \mathsf{combine}\big( \{ \mathsf{P\_Signature} \text{ } \mathit{sig} \,|\, \mathit{sig} \text{ is received in the \textsc{ack} messages}\} \big)$ \label{line:combine_ack}
\State \hskip6em \textbf{broadcast} $\langle \textsc{confirm}, h, \Sigma_h \rangle$ \label{line:broadcast_confirm}

\smallskip
\State \hskip2em every process:
\State \hskip4em \textbf{upon} reception of $\langle \textsc{confirm}, \mathsf{Hash\_Value} \text{ } h, \mathsf{T\_Signature} \text{ } \Sigma_h \rangle$ and $\mathsf{verify\_sig}(h, \Sigma_h) = \mathit{true}$: \label{line:receive_confirm}
\State \hskip6em \textbf{trigger} $\mathsf{acquire}(h, \Sigma_h)$ \label{line:acquire_spread}
\State \hskip6em \textbf{broadcast} $\langle \textsc{confirm}, h, \Sigma_h \rangle$ \label{line:echo_confirm}
\State \hskip6em \textbf{stop} executing \spread (and \sync) \label{line:stop_executing_spread}
\end{algorithmic}
\end{algorithm}

\subsection{Analysis} \label{subsection:informal_analysis}

\noindent \textbf{Correctness.}
Once all correct processes synchronize in the view $V_{\mathit{sync}}^*$ (the smallest synchronization view), all correct processes acquire a hash-signature pair.
Indeed, $\Delta = \delta\frac{n}{Y} + 3\delta$ time is sufficient for a correct leader $P_l \in \mathsf{leaders}(V_{\mathit{sync}}^*)$ to (1) disperse its proposal $\mathit{proposal}_l$ to all processes (line~\ref{line:multicast_dispersal}), (2) collect $2t + 1$ partial signatures for $h = \mathsf{hash}(\mathit{proposal}_l)$ (line~\ref{line:receive_ack}), and (3) disseminate a threshold signature for $h$ (line~\ref{line:broadcast_confirm}).
When a correct process receives the aforementioned threshold signature (line~\ref{line:receive_confirm}), it acquires the hash-signature pair (line~\ref{line:acquire_spread}) and stops executing \spread (line~\ref{line:stop_executing_spread}).

\smallskip
\noindent \textbf{Complexity.}
\spread terminates once all correct processes are synchronized in a view with a correct leader.
The synchronization is ensured in $O(\frac{n^2}{XY})$ time after GST (as $\tau_{\mathit{sync}}^* + \Delta - \text{GST} = O(\frac{n^2}{XY})$).
Hence, \spread terminates in $O(\frac{n^2}{XY})$ time after GST.

Let us analyze the number of bits \spread exchanges.
Any execution of \spread can be separated into two post-GST periods: (1) \emph{unsynchronized}, from GST until $\text{GST} + 3\delta$, and (2) \emph{synchronized}, from $\text{GST} + 3\delta$ until $\tau_{\mathit{sync}}^* + \Delta$.
First, we study the number of bits correct processes send via \textsc{dispersal}, \textsc{ack} and \textsc{confirm} message in the aforementioned periods:
\begin{compactitem}
    \item Unsynchronized period:
    Due to the $\delta$-waiting step (line~\ref{line:wait_spread}), each correct process sends \textsc{dispersal} messages (line~\ref{line:multicast_dispersal}) to (at most) $3 = O(1)$ $Y$-sized groups.
    Hence, each correct process sends $O(1) \cdot O(Y) \cdot L = O(YL)$ bits through \textsc{dispersal} messages.

    Due to the limited entrance property of \sync, each correct process enters $O(1)$ views during the unsynchronized period.
    In each view, each correct process sends (at most) $O(X)$ \textsc{ack} messages (one to each leader; line~\ref{line:send_ack}) and $O(n)$ \textsc{confirm} messages (line~\ref{line:broadcast_confirm}).
    As each \textsc{ack} and \textsc{confirm} message carries $\kappa$ bits, all correct processes send
    \begin{equation*}
        \begin{split}
        &n \cdot \big( \underbrace{O(YL)}_\text{\textsc{dispersal}} + \underbrace{O(X\kappa)}_{\text{\textsc{ack}}} + \underbrace{O(n\kappa)}_\text{\textsc{confirm}} \big) \\&= O(nYL + n^2\kappa) \text{ bits via \textsc{dispersal}, \textsc{ack} and \textsc{confirm} messages.}
        \end{split}
    \end{equation*}


    \item Synchronized period:
    Recall that all correct processes acquire a hash-signature pair (and stop executing \spread) by time $\tau_{\mathit{sync}}^* + \Delta$, and they do so in the view $V_{\mathit{sync}}^*$.
    As correct processes enter monotonically increasing views, no correct process enters a view greater than $V_{\mathit{sync}}^*$.
    
    By the stabilization property of \sync, each correct process enters a view $V \geq V_{\mathit{max}}$ by time $\text{GST} + 3\delta$.
    Moreover, until $\tau_{\mathit{sync}}^* + \Delta$, each correct process enters (at most) $O(\frac{n}{X})$ views (due to the limited synchronization view and monotonicity properties of \sync).
    Importantly, no correct leader exists in any view $V$ with $V_{\mathit{max}} < V < V_{\mathit{sync}}^*$; otherwise, $V = V_{\mathit{sync}}^*$ as processes overlap for $\Delta$ time in $V$ (due to the overlapping property of \sync).
    Hence, for each view $V$ with $V_{\mathit{max}} < V < V_{\mathit{sync}}^*$, all correct processes send $O(nX\kappa)$ bits (all through \textsc{ack} messages; line~\ref{line:send_ack}).
    In $V_{\mathit{max}}$ and $V_{\mathit{sync}}^*$, all correct processes send (1) $2 \cdot O(XnL)$ bits through \textsc{dispersal} messages (line~\ref{line:multicast_dispersal}), (2) $2 \cdot O(nX\kappa)$ bits through \textsc{ack} messages (line~\ref{line:send_ack}), and (3) $2 \cdot O(Xn\kappa)$ bits through \textsc{confirm} messages (line~\ref{line:broadcast_confirm}).
    Therefore, all correct processes send
    \begin{equation*}
        \begin{split}
        &\underbrace{O(\frac{n}{X})}_\text{$V_{\mathit{sync}}^* - V_{\mathit{max}}$} \cdot \underbrace{O(nX\kappa)}_\text{\textsc{ack}} + \underbrace{O(XnL)}_\text{\textsc{dispersal} in $V_{\mathit{max}}$ and $V_{\mathit{sync}}^*$} + \underbrace{O(nX\kappa)}_\text{\textsc{ack} in $V_{\mathit{max}}$ and $V_{\mathit{sync}}^*$} + \underbrace{O(Xn\kappa)}_\text{\textsc{confirm} in $V_{\mathit{max}}$ and $V_{\mathit{sync}}^*$} \\&= O(nXL + n^2\kappa) \text{ bits via \textsc{dispersal}, \textsc{ack} and \textsc{confirm} messages.}
        \end{split}
    \end{equation*}
\end{compactitem}
We cannot neglect the complexity of \sync, which exchanges $O(\frac{n^3\kappa}{X})$ bits during the time period $[\text{GST}, \tau_{\mathit{sync}}^* + \Delta]$.
Hence, the total number of bits \spread exchanges is
\begin{equation*}
    \underbrace{O(nYL + n^2\kappa)}_\text{unsynchronized period} + \underbrace{O(nXL + n^2\kappa)}_\text{synchronized period} + \underbrace{O(\frac{n^3\kappa}{X})}_\text{\sync} = O(nYL + nXL + \frac{n^3\kappa}{X}).
\end{equation*}
With $X = Y = \sqrt{n}$, \spread terminates in optimal $O(n)$ time, and exchanges $O(n^{1.5}L + n^{2.5}\kappa)$ bits.
Our analysis is illustrated in \Cref{fig:analysis_summary}.

\begin{figure}[ht]
   \centering
   \includegraphics[scale=0.9]{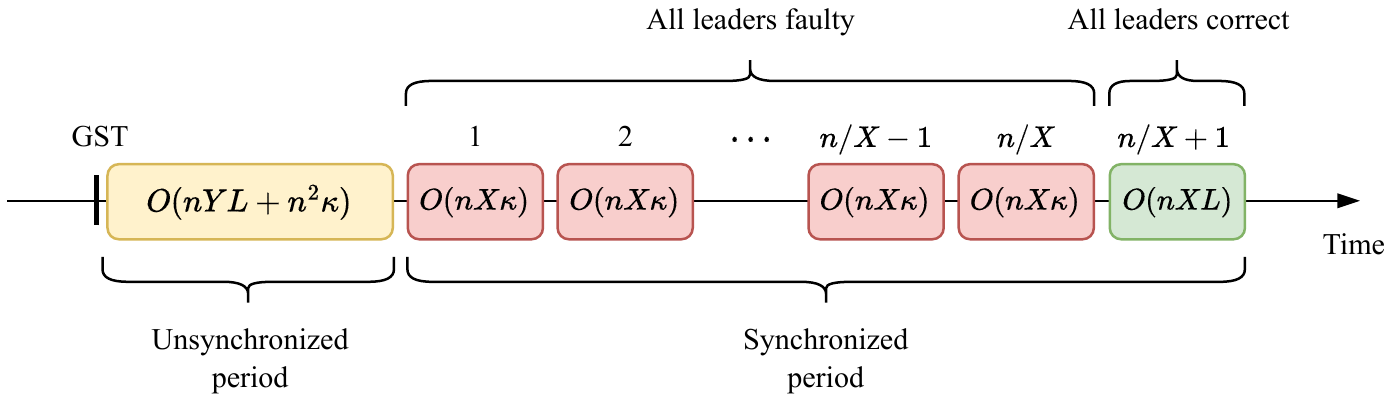}
   \caption{Illustration of \spread's bit complexity.}
   \label{fig:analysis_summary}
\end{figure}

\section{\darestark}
\label{section:darestark}

In this section, we present \darestark, a variant of \name which relies on STARK proofs.
Importantly, \darestark achieves $O(nL + n^2\mathit{poly}(\kappa))$ bit complexity, nearly tight to the $\Omega(nL + n^2)$ lower bound, while preserving optimal $O(n)$ latency.

First, we revisit \spread, pinpointing its complexity on proving RS encoding (\Cref{subsection:revisit}).
We then provide an overview on STARKs, a cryptographic primitive providing succinct proofs of knowledge (\Cref{subsection:starks}).
We finally present \darestark, which uses STARKs for provable RS encoding, thus improving on \name's complexity (\Cref{subsection:darestark_pseudocode}).

\subsection{Revisiting \name: What Causes \spread's Complexity?} \label{subsection:revisit}

Recall that \spread exchanges $O(n^{1.5}L + n^{2.5}\kappa)$ bits.
This is due to a fundamental requirement of \rec: at least $t + 1$ correct processes must have obtained the value $v$ by the time \q decides $h = \mathsf{hash}(v)$.
\rec leverages this requirement to prove the correct encoding of RS symbols.
In brief (as explained in \Cref{subsubsection:rec_overview}): (1) every correct process $P$ that obtained $v \neq \bot$ encodes it in $n$ RS symbols $s_1, \ldots, s_n$; (2) $P$ sends each $s_i$ to $P_i$; (3) upon receiving $t + 1$ identical copies of $s_i$, $P_i$ can trust $s_i$ to be the $i$-th RS symbol for $v$ (note that $s_i$ can be trusted only because it was produced by at least one correct process - nothing else proves $s_i$'s relationship to $v$!); (4) every correct process $P_i$ disseminates $s_i$, enabling the reconstruction of $v$ by means of error-correcting decoding.
In summary, \name bottlenecks on \spread, and \spread's complexity is owed to the need to prove the correct encoding of RS symbols in \rec.
Succinct arguments of knowledge (such as STARKs), however, allow to publicly prove the relationship between an RS symbol and the value it encodes, eliminating the need to disperse the entire value to $t + 1$ correct processes -- a dispersal of provably correct RS symbols suffices.
\darestark builds upon this idea.




\subsection{STARKs} \label{subsection:starks}

First introduced in \cite{ben2018scalable}, STARKs are succinct, universal, transparent arguments of knowledge.
For any function $f$ (computable in polynomial time) and any (polynomially-sized) $y$, a STARK can be used to prove the knowledge of some $x$ such that $f(x) = y$.
Remarkably, the size of a STARK proof is $O(\mathit{poly}(\kappa))$.
At a very high level, a STARK proof is produced as follows: (1) the computation of $f(x)$ is unfolded on an execution trace; (2) the execution trace is (RS) over-sampled for error amplification; (3) the correct computation of $f$ is expressed as a set of algebraic constraints over the trace symbols; (4) the trace symbols are organized in a Merkle tree~\cite{merkle-tree-crypto87}; (5) the tree's root is used as a seed to pseudo-randomly sample the trace symbols. The resulting collection of Merkle proofs proves that, for some known (but not revealed) $x$, $f(x) \neq y$ only with cryptographically low probability (negligible in $\kappa$).
STARKs are non-interactive, require no trusted setup (they are transparent), and their security reduces to that of cryptographic hashes in the Random Oracle Model (ROM) \cite{bellare1993}.

\subsection{Implementation} \label{subsection:darestark_pseudocode}

\noindent \textbf{Provably correct encoding.}
At its core, \darestark uses STARKs to attest the correct RS encoding of values.
For every $i \in [1, n]$, we define $\mathsf{shard}_i(\cdot)$ by
\begin{equation}
\label{equation:starkshard}
    \mathsf{shard}_i(v \in \mathsf{Value}) = \begin{cases}
        \big( \mathsf{hash}(v), \mathsf{encode}_i(v) \big), &\text{if and only if}\; \mathsf{valid}(v) = \mathit{true} \\
        \bot, &\text{otherwise,}
    \end{cases}
\end{equation}
where $\mathsf{encode}_i(v)$ represents the $i$-th RS symbol obtained from $\mathsf{encode}(v)$ (see \Cref{section:preliminaries}).
We use $\mathsf{proof}_i(v)$ to denote the STARK proving the correct computation of $\mathsf{shard}_i(v)$.
The design and security of \darestark rests on the following theorem.

\begin{theorem}
\label{theorem:starkshard}
    Let $i_1, \ldots, i_{t + 1}$ be distinct indices in $[1, n]$. Let $h$ be a hash, let $s_1, \ldots, s_{t + 1}$ be RS symbols, let $\mathit{stark}_1, \ldots, \mathit{stark}_{t + 1}$ be STARK proofs such that, for every $k \in [1, t + 1]$, $\mathit{stark}_k$ proves knowledge of some (undisclosed) $v_k$ such that $\mathsf{shard}_{i_k}(v_k) = (h, s_k)$. We have that
    \begin{equation*}
        v = \mathsf{decode}(\{s_1, \ldots, s_k\})
    \end{equation*}
    satisfies $\mathsf{valid}(v) = \mathit{true}$ and $\mathsf{hash}(v) = h$.
\end{theorem}

    \begin{proof}
        For all $k$, by the correctness of $\mathit{stark}_k$ and \cref{equation:starkshard}, we have that (1) $h = \mathsf{hash}(v_k)$, (2) $s_k = \mathsf{encode}_{i_k}(v_k)$, and (3) $\mathsf{valid}(v_k) = \mathit{true}$.
        By the collision-resistance of $\mathsf{hash}(\cdot)$, for all $k, k'$, we have $v_k = v_k'$.
        By the definition of $\mathsf{encode}(\cdot)$ and $\mathsf{decode}(\cdot)$, we then have
        \begin{equation*}
            v = \mathsf{decode}(\{s_1, \ldots, s_k\}) = v_1 = \ldots = v_{t + 1},
        \end{equation*}
        which implies that $\mathsf{valid}(v) = \mathit{true}$ and $\mathsf{hash}(v) = h$.
    \end{proof}

\begin{algorithm} [ht]
\caption{\darestark: Pseudocode (for process $P_i$)}
\label{algorithm:darestark}
\begin{algorithmic} [1]
\footnotesize
\State \textbf{Uses:}
\State \hskip2em \q, \textbf{instance} $\mathit{agreement}$

\smallskip
\State \textbf{upon} $\mathsf{init}$:
\State \hskip2em $\mathsf{Hash\_Value}$ $\mathit{proposed\_hash}_i \gets \bot$
\State \hskip2em $\mathsf{Map}\big( \mathsf{Hash\_Value} \to (\mathsf{Symbol}, \mathsf{STARK}) \big)$ $\mathit{proposal\_shards}_i \gets \text{empty}$
\State \hskip2em $\mathsf{Map}\big( \mathsf{Hash\_Value} \to \mathsf{Set}(\mathsf{Symbol}) \big)$ $\mathit{decision\_symbols}_i \gets \text{empty}$
\State \hskip2em $\mathsf{Bool}$ $\mathit{decided}_i \gets \mathit{false}$

\smallskip
\State \textcolor{blue}{\(\triangleright\) Dispersal}
\State \textbf{upon} $\mathsf{propose}(\mathsf{Value} \text{ } v_i)$: \label{line:stark_propose}
\State \hskip2em $\mathit{proposed\_hash}_i \gets \mathsf{hash}(v_i)$
\State \hskip2em \textbf{for} $\mathsf{Integer} \text{ } k \gets 1 \text{ to } n$:
\State \hskip4em $(\mathsf{Hash} \text{ } h_k, \mathsf{Symbol} \text{ } s_k) \gets \mathsf{shard}_k(v_i)$ \label{line:compute_shard}
\State \hskip4em $\mathsf{STARK}$ $\mathit{stark}_k \gets \mathsf{proof}_k(v_i)$ \label{line:compute_stark}
\State \hskip4em \textbf{send} $\langle \textsc{dispersal}, h_k, s_k, \mathit{stark}_k \rangle$ to $P_k$ \label{line:stark_send_dispersal}

\smallskip
\State \textbf{upon} reception of $\langle \textsc{dispersal}, \mathsf{Hash\_Value} \text{ } h, \mathsf{Symbol} \text{ } s, \mathsf{STARK} \text{ } \mathit{stark} \rangle$ from process $P_j$ and $\mathit{stark}$ proves $\mathsf{shard}_i(?) = (h, s)$: \label{line:stark_receive_shard}
\State \hskip2em $\mathit{proposal\_shards}_i[h] \gets (s, \mathit{stark})$ \label{line:stark_store_shard}
\State \hskip2em \textbf{send} $\langle \textsc{ack}, \mathsf{share\_sign}_i(h) \rangle$ to $P_j$
\label{line:stark_send_ack}

\smallskip
\State \textcolor{blue}{\(\triangleright\) Agreement}
\State \textbf{upon} $\langle \textsc{ack}, \mathsf{P\_Signature} \text{ } \mathit{sig} \rangle$ is received from $2t + 1$ processes: \label{line:stark_received_acks}
\State \hskip2em $\mathsf{T\_Signature}$ $\Sigma \gets \mathsf{combine}\big( \{\mathit{sig} \,|\, \mathit{sig} \text{ is received in the \textsc{ack} messages}\} \big)$ \label{line:combine_ack_2}
\State \hskip2em \textbf{invoke} $\mathit{agreement}.\mathsf{propose}(\mathit{proposed\_hash}_i, \Sigma)$ \label{line:stark_agreement_propose}

\smallskip
\State \textbf{upon} $\mathit{agreement}.\mathsf{decide}(\mathsf{Hash\_Value} \text{ } h, \mathsf{T\_Signature} \text{ } \Sigma)$ with $\mathit{proposal\_shards}_i[h] \neq \bot$: \label{line:stark_q_decide}
\State \hskip2em $(\mathsf{Symbol} \text{ } s, \mathsf{STARK} \text{ } \mathit{stark}) \gets \mathit{proposal\_shards}_i[h]$
\State \hskip2em \textbf{broadcast} $\langle \textsc{retrieve}, h, s, \mathit{stark} \rangle$ \label{line:stark_broadcast_retrieval}

\smallskip
\State \textcolor{blue}{\(\triangleright\) Retrieval}
\State \textbf{upon} reception of $\langle \textsc{retrieve}, \mathsf{Hash\_Value} \text{ } h, \mathsf{Symbol} \text{ } s, \mathsf{STARK} \text{ } \mathit{stark} \rangle$ from process $P_j$ and $\mathit{stark}$ proves $\mathsf{shard}_j(?) = (h, s)$: \label{line:stark_check_sample}
\State \hskip2em $\mathit{decision\_symbols}_i[h] \gets \mathit{decision\_symbols}_i[h] \cup \{s\}$

\smallskip
\State \textbf{upon} (1) exists $\mathsf{Hash\_Value}$ $h$ such that $\mathit{decision\_symbols}_i[h]$ has $t + 1$ elements, and (2) $\mathit{decided}_i = \mathit{false}$: \label{line:stark_decide_rule}
\State \hskip2em $\mathit{decided}_i \gets \mathit{true}$
\State \hskip2em \textbf{trigger} $\mathsf{decide}\big( \mathsf{decode}(\mathit{decision\_symbols}_i[h]) \big)$ \label{line:stark_decide}

\end{algorithmic}
\end{algorithm}

\smallskip
\noindent \textbf{Algorithm description.}
The pseudocode of \darestark is presented in \Cref{algorithm:darestark} from the perspective of a correct process $P_i$.
Similarly to \name, \darestark unfolds in three phases:
\begin{compactenum}
    \item \emph{Dispersal:} Upon proposing a value $v_i$ (line~\ref{line:stark_propose}), $P_i$ sends (line~\ref{line:stark_send_dispersal}) to each process $P_k$ (1) $(h_k, s_k) = \mathsf{shard}_k(v_i)$ (computed at line~\ref{line:compute_shard}), and (2) $\mathit{stark}_k = \mathsf{proof}_k(v_i)$ (computed at line~\ref{line:compute_stark}).
    In doing so (see \Cref{theorem:starkshard}), $P_i$ proves to $P_k$ that $h_k = \mathsf{hash}(v_i)$ is the hash of a valid proposal, whose $k$-th RS symbol is $\mathsf{encode}_k(v_i)$.
    $P_k$ checks $\mathit{stark}_k$ against $(h_k, s_k)$ (line~\ref{line:stark_receive_shard}), stores $(s_k, \mathit{stark}_k)$ (line~\ref{line:stark_store_shard}), and sends a partial signature for $h_k$ back to $P_i$ (line~\ref{line:stark_send_ack}).

    \item \emph{Agreement:} Having collected a threshold signature $\Sigma$ for $\mathsf{hash}(v_i)$ (line~\ref{line:combine_ack_2}), $P_i$ proposes $(\mathsf{hash}(v_i), \Sigma)$ to \q (line~\ref{line:stark_agreement_propose}).

    \item \emph{Retrieval:}
    Upon deciding a hash $h$ from \q (line~\ref{line:stark_q_decide}), $P_i$ broadcasts (if available) the $i$-th RS symbol for $h$, along with the relevant proof (line~\ref{line:stark_broadcast_retrieval}).
    Upon receiving $t + 1$ symbols $S$ for the same hash (line~\ref{line:stark_decide_rule}), $P_i$ decides $\mathsf{decode}(S)$ (line~\ref{line:stark_decide}).
\end{compactenum}

\smallskip
\noindent \textbf{Analysis.}
Upon proposing a value $v_i$ (line~\ref{line:stark_propose}), a correct process $P_i$ sends $\mathsf{shard}_k(v_i)$ and $\mathsf{proof}_k(v_i)$ to each process $P_k$ (line~\ref{line:stark_send_dispersal}).
Checking $\mathsf{proof}_k(v_i)$ against $\mathsf{shard}_k(v_i)$ (line~\ref{line:stark_receive_shard}), $P_k$ confirms having received the $k$-th RS symbol for $v_i$ (note that this does not require the transmission of $v_i$, just $\mathsf{hash}(v_i)$).
As $2t + 1$ processes are correct, $P_i$ is guaranteed to eventually gather a $(2t + 1)$-threshold signature $\Sigma$ for $\mathsf{hash}(v_i)$ (line~\ref{line:combine_ack_2}). 
Upon doing so, $P_i$ proposes $(\mathsf{hash}(v_i), \Sigma)$ to \q (line~\ref{line:stark_agreement_propose}).
Since every correct process eventually proposes a value to \q, every correct process eventually decides some hash $h$ from \q (line~\ref{line:stark_q_decide}).
Because $2t + 1$ processes signed $h$, at least $t + 1$ correct processes (without loss of generality, $P_1, \ldots, P_{t + 1}$) received a correctly encoded RS-symbol for $h$.
More precisely, for every $k \in [1, t + 1]$, $P_k$ received and stored the $k$-th RS symbol encoded from the pre-image $v$ of $h$.
Upon deciding from \q, each process $P_k$ broadcasts its RS symbol, along with the relevant proof (line~\ref{line:stark_broadcast_retrieval}).
Because at most $t$ processes are faulty, no correct process receives $t + 1$ RS symbols pertaining to a hash other than $h$.
As $P_1, \ldots, P_{t + 1}$ all broadcast their symbols and proofs, eventually every correct process collects $t + 1$ (provably correct) RS symbols $S$ pertaining to $h$ (line~\ref{line:stark_decide_rule}), and decides $\mathsf{decode}(S)$ (line~\ref{line:stark_decide}).
By \Cref{theorem:starkshard}, every correct process eventually decides the same valid value $v$ (with $h = \mathsf{hash}(v)$).

Concerning bit complexity, throughout an execution of \darestark, a correct process engages once in \q (which exchanges $O(n^2\kappa)$ bits in total) and sends: (1) $n$ \textsc{dispersal} messages, each of size $O(\frac{L}{n} + \mathit{poly}(\kappa))$, (2) $n$ \textsc{ack} messages, each of size $O(\kappa)$, and (3) $n$ \textsc{retrieve} messages, each of size $O(\frac{L}{n} + \mathit{poly}(\kappa))$. 
Therefore, the bit complexity of \darestark is $O(nL + n^2\mathit{poly}(\kappa))$.
As for the latency, it is $O(n)$ (due to the linear latency of \q).
\section{Concluding Remarks} \label{section:conclusion}

This paper introduces \name (Disperse, Agree, REtrieve), the first partially synchronous Byzantine agreement algorithm on values of $L$ bits with better than $O(n^2L)$ bit complexity and sub-exponential latency.
\name achieves $O(n^{1.5}L + n^{2.5}\kappa)$ bit complexity ($\kappa$ is the security parameter) and optimal $O(n)$ latency, which is an effective $\sqrt{n}$ factor bit-improvement for $L \geq n\kappa$ (typical in practice).
\name achieves its complexity in two steps.
First, \name decomposes problem of agreeing on large values ($L$ bits) into three sub-problems: (1) value dispersal, (2) validated agreement on small values ($O(\kappa)$), and (3) value retrieval.
(\name effectively acts as an extension protocol for Byzantine agreement.)
Second, \name's novel dispersal algorithm solves the main challenge, value dispersal, using only $O(n^{1.5}L)$ bits and linear latency.

Moreover, we prove that the lower bound of $\Omega(nL + n^2)$ is near-tight by matching it near-optimally with \darestark, a modified version of \name using STARK proofs that reaches $O(nL + n^2\mathit{poly}(\kappa))$ bits and maintains optimal $O(n)$ latency.
We hope \darestark motivates research into more efficient STARK schemes in the future, which currently have large hidden constants affecting their practical use.



\bibstyle{plainurl}
\bibliography{references}

\appendix
\section{\spread: Complete Pseudocode \& Proof}
\label{appendix:spread_proof}

In this section, we give a complete pseudocode of \spread, and formally prove its correctness and complexity.
First, we focus on \sync: we present \sync's implementation, as well as its proof of correctness and complexity (\Cref{subsection:sync_pseudocode}).
Then, we prove the correctness and complexity of \spread (\Cref{subsection:spread_proof}).
Throughout the entire section, $X = Y = \sqrt{n}$.
For simplicity, we assume that $\sqrt{n}$ is an integer.

\subsection{\sync: Implementation \& Proof} \label{subsection:sync_pseudocode}

\noindent \textbf{Algorithm description.}
The pseudocode of \sync is given in \Cref{algorithm:sync}, and it highly resembles \textsc{RareSync}~\cite{CivitDGGGKV22}.
Each process $P_i$ has an access to two timers: (1) $\mathit{view\_timer}_i$, and (2) $\mathit{dissemination\_timer}_i$.
A timer exposes the following interface:
\begin{compactitem}
    \item $\mathsf{measure}(\mathsf{Time} \text{ } x)$: After exactly $x$ time as measured by the local clock, the timer expires (i.e., an expiration event is received by the host).
    As local clocks can drift before GST, timers might not be precise before GST: $x$ time as measured by the local clock may not amount to $x$ real time.

    \item $\mathsf{cancel}()$:
    All previously invoked $\mathsf{measure}(\cdot)$ methods (on that timer) are cancelled.
    Namely, all pending expiration events (associated with that timer) are removed from the event queue.
\end{compactitem}

We now explain how \sync works, and we do so from the perspective of a correct process $P_i$.
When $P_i$ starts executing \sync (line~\ref{line:start_synchronizer}), it starts measuring $\mathit{view\_duration} = \Delta + 2\delta$ time using $\mathit{view\_timer}_i$ (line~\ref{line:measure_first}) and enters the first view (line~\ref{line:enter_first_view}).

Once $\mathit{view\_timer}_i$ expires (line~\ref{line:view_expires}), which signals that $P_i$'s current view has finished, $P_i$ notifies every process of this via a \textsc{view-completed} message (line~\ref{line:broadcast_view_completed}).
When $P_i$ learns that $2t + 1$ processes have completed some view $V \geq \mathit{view}_i$ (line~\ref{line:receive_view_completed}) or any process started some view $V' > \mathit{view}_i$ (line~\ref{line:receive_enter_view}), $P_i$ prepares to enter a new view (either $V + 1$ or $V'$).
Namely, $P_i$ (1) cancels $\mathit{view\_timer}_i$ (line~\ref{line:cancel_view_1} or line~\ref{line:cancel_view_2}), (2) cancels $\mathit{dissemination\_timer}_i$ (line~\ref{line:cancel_dissemination_timer_1} or line~\ref{line:cancel_dissemination_timer_2}), and (3) starts measuring $\delta$ time using $\mathit{dissemination\_timer}_i$ (line~\ref{line:measure_dissemination_timer_1} or line~\ref{line:measure_dissemination_timer_2}).
Importantly, $P_i$ measures $\delta$ time (using $\mathit{dissemination\_timer}_i$) before entering a new view in order to ensure that $P_i$ enters only $O(1)$ views during the time period $[\text{GST}, \text{GST} + 3\delta)$.
Finally, once $\mathit{dissemination\_timer}_i$ expires (line~\ref{line:dissemination_timer_expire}), $P_i$ enters a new view (line~\ref{line:enter_v}).
 
\begin{algorithm} [t]
\caption{\sync: Pseudocode (for process $P_i$)}
\label{algorithm:sync}
\begin{algorithmic} [1]
\footnotesize

\State \textbf{Uses:}
\State \hskip2em $\mathsf{Timer}$ $\mathit{view\_timer}_i$
\State \hskip2em $\mathsf{Timer}$ $\mathit{dissemination\_timer}_i$

\smallskip
\State \textbf{Functions:}
\State \hskip2em $\mathsf{leaders}(\mathsf{View} \text{ } V) \equiv \{P_{\big( (V \mathsf{mod} \sqrt{n}) - 1\big)\sqrt{n} + 1}, P_{\big( (V \mathsf{mod} \sqrt{n}) - 1\big)\sqrt{n} + 2}, ..., P_{\big( (V \mathsf{mod} \sqrt{n}) - 1\big)\sqrt{n} + \sqrt{n}}\}$ \label{line:round_robin}

\smallskip
\State \textbf{Constants:}
\State \hskip2em $\mathsf{Time}$ $\mathit{view\_duration} = \Delta + 2\delta = \delta\sqrt{n} + 3\delta + 2\delta$

\smallskip
\State \textbf{Variables:}
\State \hskip2em $\mathsf{View}$ $\mathit{view}_i \gets 1$
\State \hskip2em $\mathsf{T\_Signature}$ $\mathit{view\_sig}_i \gets \bot$

\smallskip
\State \textbf{upon} $\mathsf{init}$: \BlueComment{start of the algorithm} \label{line:start_synchronizer}
\State \hskip2em $\mathit{view\_timer}_i.\mathsf{measure}(\mathit{view\_duration})$ \label{line:measure_first}
\State \hskip2em \textbf{trigger} $\mathsf{advance}(1)$ \label{line:enter_first_view}

\smallskip
\State \textbf{upon} $\mathit{view\_timer}_i$ \textbf{expires:} \label{line:view_expires}
\State \hskip2em \textbf{broadcast} $\langle \textsc{view-completed}, \mathit{view}_i, \mathsf{share\_sign}_i(\mathit{view}_i) \rangle$\label{line:broadcast_view_completed}

\smallskip
\State \textbf{upon} exists $\mathsf{View}$ $V \geq \mathit{view}_i$ with $2t + 1$ $\langle \textsc{view-completed}, V, \mathsf{P\_Signature} \text{ } \mathit{sig} \rangle$ received messages: \label{line:receive_view_completed}
\State \hskip2em $\mathit{view\_sig}_i \gets \mathsf{combine}\big( \{ \mathit{sig} \,|\, \mathit{sig} \text{ is received in the } \textsc{view-completed} \text{ messages}  \} \big)$
\State \hskip2em $\mathit{view}_i \gets V + 1$ \label{line:update_view_view_completed}
\State \hskip2em $\mathit{view\_timer}_i.\mathsf{cancel()}$ \label{line:cancel_view_1}
\State \hskip2em $\mathit{dissemination\_timer}_i.\mathsf{cancel()}$ \label{line:cancel_dissemination_timer_1}
\State \hskip2em $\mathit{dissemination\_timer}_i.\mathsf{measure}(\delta)$ \label{line:measure_dissemination_timer_1}

\smallskip
\State \textbf{upon} reception of $\langle \textsc{enter-view}, \mathsf{View} \text{ } V, \mathsf{T\_Signature} \text{ } \mathit{sig} \rangle$ with $V > \mathit{view}_i$: \label{line:receive_enter_view}
\State \hskip2em $\mathit{view\_sig}_i \gets \mathit{sig}$
\State \hskip2em $\mathit{view}_i \gets V$ \label{line:update_view_enter_view}
\State \hskip2em $\mathit{view\_timer}_i.\mathsf{cancel()}$ \label{line:cancel_view_2}
\State \hskip2em $\mathit{dissemination\_timer}_i.\mathsf{cancel()}$ \label{line:cancel_dissemination_timer_2}
\State \hskip2em $\mathit{dissemination\_timer}_i.\mathsf{measure}(\delta)$ \label{line:measure_dissemination_timer_2}

\smallskip
\State \textbf{upon} $\mathit{dissemination\_timer}_i$ \textbf{expires:} \label{line:dissemination_timer_expire}
\State \hskip2em \textbf{broadcast} $\langle \textsc{enter-view}, \mathit{view}_i, \mathit{view\_sig}_i \rangle$ \label{line:broadcast_enter_epoch}
\State \hskip2em $\mathit{view\_timer}_i.\mathsf{measure}(\mathit{view\_duration})$ \label{line:measure_second}
\State \hskip2em \textbf{trigger} $\mathsf{advance}(\mathit{view}_i)$ \label{line:enter_v}

\end{algorithmic}
\end{algorithm}

\smallskip
\noindent \textbf{Proof of correctness \& complexity.}
We prove that \sync satisfies eventual synchronization and all the properties on which \spread relies (recall \Cref{subsection:spread_algorithm}).

Given a view $V$, $\tau_V$ denotes the first time a correct process enters $V$.
First, we prove that, if any view $V > 1$ is entered by a correct process, then view $V - 1$ was previously entered by a correct process.

\begin{lemma} \label{lemma:entering}
Consider any view $V > 1$ which is entered by a correct process (line~\ref{line:enter_v}).
View $V - 1$ is entered by a correct process by time $\tau_V$.
\end{lemma}
\begin{proof}
As $V > 1$ is entered by a correct process at time $\tau_V$, at least $t + 1$ correct processes have broadcast a \textsc{view-completed} message for $V - 1$ (line~\ref{line:broadcast_view_completed}) by time $\tau_V$.
Therefore, any such correct process $P_i$ has previously invoked the $\mathsf{measure}(\cdot)$ operation on $\mathit{view\_timer}_i$ (line~\ref{line:measure_first} or line~\ref{line:measure_second}) at some time $\tau_i \leq \tau_V$.
As this specific invocation of the $\mathsf{measure}(\cdot)$ operation is not cancelled (otherwise, $P_i$ would not broadcast a \textsc{view-completed} message because of this invocation), $\mathit{view}_i = V - 1$ at $\tau_i$.
Immediately after the invocation (still at time $\tau_i$ as local computation takes zero time), $P_i$ enters $\mathit{view}_i = V - 1$ (line~\ref{line:enter_first_view} or line~\ref{line:enter_v}), which concludes the proof.
\end{proof}

A direct consequence of Lemma~\ref{lemma:entering} is that $\tau_{V'} \geq \tau_{V}$ for any two views $V'$ and $V$ with $V' > V$.
Next, we prove that no correct process updates its $\mathit{view}_i$ variable to a view $V' > V$ before time $\tau_V + \mathit{view\_duration}$ if $\tau_V \geq \text{GST}$.

\begin{lemma} \label{lemma:view_updated}
Consider any view $V$ such that $\tau_V \geq \text{GST}$.
No correct process $P_i$ updates its $\mathit{view}_i$ variable to a view $V' > V$ before time $\tau_V + \mathit{view\_duration}$.
\end{lemma}
\begin{proof}
No correct process sends a \textsc{view-completed} message (line~\ref{line:broadcast_view_completed}) for $V$ (or, by Lemma~\ref{lemma:entering}, any greater view) before time $\tau_V + \mathit{view\_duration}$ (as local clocks do not drift after GST).
Thus, before time $\tau_V + \mathit{view\_duration}$, no correct process $P_i$ receives (1) $2t + 1$ \textsc{view-completed} messages for $V$ or any greater view (line~\ref{line:receive_view_completed}), nor (2) an \textsc{enter-view} message for a view greater than $V$ (line~\ref{line:receive_enter_view}), which means that $\mathit{view}_i$ cannot be updated before time $\tau_V + \mathit{view\_duration}$.
\end{proof}

The next lemma proves that every correct process enters $V$ by time $\tau_V + 2\delta$, for any view $V > V_{\mathit{max}}$.
Recall that $V_{\mathit{max}}$ is the greatest view entered by a correct process before GST.

\begin{lemma} \label{lemma:synchronized_start}
For any view $V > V_{\mathit{max}}$, every correct process enters $V$ by time $\tau_V + 2\delta$.
\end{lemma}
\begin{proof}
If $V = 1$ (which implies that $V_{\mathit{max}} = 0$, i.e., no correct process started executing \sync before GST), every correct process enters $V$ at GST.
Hence, the statement of the lemma holds.

Let $V > 1$.
Due to the definition of $V_{\mathit{max}}$, $\tau_V \geq \text{GST}$.
Every correct process $P_i$ receives an \textsc{enter-view} message for $V$ (line~\ref{line:receive_enter_view}) by time $\tau_V + \delta$.
When $P_i$ receives the aforementioned message, it updates $\mathit{view}_i$ to $V$ (line~\ref{line:update_view_enter_view}).
Moreover, $P_i$ does not update its $\mathit{view}_i$ variable before time $\tau_V + \mathit{view\_duration} > \tau_V + 2\delta$ (by Lemma~\ref{lemma:view_updated}).
Hence, $\mathit{dissemination\_timer}_i$ expires (line~\ref{line:dissemination_timer_expire}), and $P_i$ enters $V$ by time $\tau_V + 2\delta$ (line~\ref{line:enter_v}).
\end{proof}

The following lemma proves that \sync ensures the overlapping property.

\begin{lemma} [Overlapping] \label{lemma:overlapping}
For any view $V > V_{\mathit{max}}$, all correct processes overlap in $V$ for (at least) $\Delta$ time.
\end{lemma}
\begin{proof}
For any view $V > V_{\mathit{max}}$, $\tau_V \geq \text{GST}$ due to the definition of $V_{\mathit{max}}$.
The following holds:
\begin{compactitem}
    \item Every correct process enters $V$ by time $\tau_V + 2\delta$ (by Lemma~\ref{lemma:synchronized_start}).

    \item No correct process enters another view before time $\tau_V + \mathit{view\_duration}$ (by Lemma~\ref{lemma:view_updated}).
\end{compactitem}
Thus, all correct process are in $V$ from $\tau_V + 2\delta$ until $\tau_V + \mathit{view\_duration}$, for any view $V > V_{\mathit{max}}$.
In other words, in any view $V > V_{\mathit{max}}$, all correct processes overlap for (at least) $\mathit{view\_duration} - 2\delta = \Delta$ time.
\end{proof}

We are ready to prove that \sync satisfies eventual synchronization.

\begin{theorem}
\sync is correct, i.e., it satisfies eventual synchronization.
\end{theorem}
\begin{proof}
By Lemma~\ref{lemma:overlapping}, all correct processes overlap for (at least) $\mathit{view\_duration} - 2\delta = \Delta$ time in every view $V > V_{\mathit{max}}$.
Given that the leaders are assigned in the round-robin manner (line~\ref{line:round_robin}), there exists a view $V_{\mathit{sync}} > V_{\mathit{max}}$ with a correct leader, which concludes the proof. 
\end{proof}

In the rest of the subsection, we prove that \sync satisfies other properties specified in \Cref{subsection:spread_algorithm}.
We start by proving the monotonicity property of \sync.

\begin{lemma} [Monotonicity]
Any correct process enters monotonically increasing views.
\end{lemma}
\begin{proof}
Consider any correct process $P_i$.
The lemma follows as the value of the $\mathit{view}_i$ variable only increases (line~\ref{line:update_view_view_completed} or line~\ref{line:update_view_enter_view}).
\end{proof}

Next, we show a direct consequence of Lemma~\ref{lemma:view_updated}: no correct process $P_i$ updates its $\mathit{view}_i$ variable to a view greater than $V_{\mathit{max}} + 1$ by time $\text{GST} + 3\delta$.

\begin{lemma} \label{lemma:view_v_max_1}
No correct process $P_i$ updates its $\mathit{view}_i$ variable to any view greater than $V_{\mathit{max}} + 1$ by time $\text{GST} + 3\delta$.
\end{lemma}
\begin{proof}
By the definition of $V_{\mathit{max}}$, no correct process enters $V_{\mathit{max}} + 1$ (or a greater view) before GST.
Moreover, $\mathit{view\_duration} > 3\delta$.
Thus, the lemma follows from Lemma~\ref{lemma:view_updated}.
\end{proof}

Next, we prove the stabilization property of \sync.

\begin{lemma} [Stabilization] \label{lemma:all_enter}
Any correct process enters a view $V \geq V_{\mathit{max}}$ by time $\text{GST} + 3\delta$.
\end{lemma}
\begin{proof}
If $V_{\mathit{max}} = 0$ (no correct process started executing \sync before GST), then all correct processes enter view $V_{\mathit{max}} + 1 = 1$ at GST.
Thus, the lemma holds in this case.

Let $V_{\mathit{max}} > 0$.
Consider a correct process $P_i$.
By time $\text{GST} + \delta$, $P_i$ receives an \textsc{enter-view} message for $V_{\mathit{max}}$ (line~\ref{line:receive_enter_view}) as this message was previously broadcast (line~\ref{line:broadcast_enter_epoch}).
Hence, at time $\text{GST} + \delta$, $\mathit{view}_i$ is either $V_{\mathit{max}}$ or $V_{\mathit{max}} + 1$; note that it cannot be greater than $V_{\mathit{max}} + 1$ by Lemma~\ref{lemma:view_v_max_1}.
We consider two cases:
\begin{compactitem}
    \item $\mathit{view}_i = V_{\mathit{max}} + 1$ at time $\text{GST} + \delta$:
    In this case, $P_i$ does not cancel $\mathit{dissemination\_timer}_i$ (line~\ref{line:cancel_dissemination_timer_1} or line~\ref{line:cancel_dissemination_timer_2}) before it expires (by Lemma~\ref{lemma:view_v_max_1}).
    Therefore, $P_i$ enters $V_{\mathit{max}} + 1$ (line~\ref{line:enter_v}) by time $\text{GST} + 2\delta$ (after $\mathit{dissemination\_timer}_i$ expires at line~\ref{line:dissemination_timer_expire}).
    The claim of the lemma holds.

    \item $\mathit{view}_i = V_{\mathit{max}}$ at time $\text{GST} + \delta$:
    If $P_i$ does not update its $\mathit{view}_i$ variable by time $\text{GST} + 2\delta$, $\mathit{dissemination\_timer}_i$ expires (line~\ref{line:dissemination_timer_expire}), and $P_i$ enters $V_{\mathit{max}}$ (line~\ref{line:enter_v}) by time $\text{GST} + 2\delta$.
    
    Otherwise, $P_i$ updates its $\mathit{view}_i$ variable to $V_{\mathit{max}} + 1$ (by Lemma~\ref{lemma:view_v_max_1}), and measures $\delta$ time using $\mathit{dissemination\_timer}_i$ (line~\ref{line:measure_dissemination_timer_1} or line~\ref{line:measure_dissemination_timer_2}).
    Moreover, Lemma~\ref{lemma:view_v_max_1} shows that $P_i$ does cancel $\mathit{dissemination\_timer}_i$.
    Hence, when $\mathit{dissemination\_timer}_i$ expires (which happens by time $\text{GST} + 3\delta$ at line~\ref{line:dissemination_timer_expire}), $P_i$ enters $V_{\mathit{max}} + 1$ (line~\ref{line:enter_v}).
\end{compactitem}
The lemma is true as it holds in both possible cases.    
\end{proof} 

Next, we show that \sync satisfies the limited entrance property.

\begin{lemma} [Limited entrance] \label{lemma:5_delta}
In the time period $[\text{GST}, \text{GST} + 3\delta)$, any correct process enters $O(1)$ views.
\end{lemma}
\begin{proof}
The lemma holds as any correct process waits $\delta$ time (line~\ref{line:measure_dissemination_timer_1} or line~\ref{line:measure_dissemination_timer_2}) before entering a new view (local clocks do not drift after GST).
\end{proof}

Next, we prove the limited synchronization view property of \sync.
Recall that $V_{\mathit{sync}}^*$ is the smallest synchronization view.

\begin{lemma} [Limited synchronization view] \label{lemma:views_difference}
$V_{\mathit{sync}}^* - V_{\mathit{max}} = O(\sqrt{n})$.
\end{lemma}
\begin{proof}
By Lemma~\ref{lemma:overlapping}, all correct processes overlap in any view $V > V_{\mathit{max}}$ for (at least) $\Delta$ time.
Consider the set of views $\mathbb{V} = \{V_{\mathit{max}} + 1, V_{\mathit{max}} + 2, ..., V_{\mathit{max}} + \sqrt{n}\}$.
Each process is a leader in exactly one view from $\mathbb{V}$ (due to the round-robin rotation of the leaders).
Hence, $V_{\mathit{sync}}^* \in \mathbb{V}$, which concludes the proof.
\end{proof}



Finally, we prove the complexity of \sync.

\begin{lemma} [Complexity]
\sync exchanges $O(n^{2.5}\kappa)$ bits during the time period $[\text{GST}, \tau_{\mathit{sync}}^* + \Delta]$, and it synchronizes all correct processes within $O(n)$ time after GST ($\tau_{\mathit{sync}}^* + \Delta - \text{GST} = O(n)$)
\end{lemma}
\begin{proof}
Due to the definition of $V_{\mathit{sync}}^*$, \sync synchronizes all correct processes in $V_{\mathit{sync}}^*$.
For each view, each correct process sends $O(n\kappa)$ bits via the \textsc{view-completed} (line~\ref{line:broadcast_view_completed}) and \textsc{enter-view} (line~\ref{line:broadcast_enter_epoch}) messages.
By Lemma~\ref{lemma:5_delta}, each correct process $P_i$ enters $O(1)$ views during the time period $[\text{GST}, \text{GST} + 3\delta)$.
Moreover, Lemma~\ref{lemma:all_enter} shows that all correct processes enter a view $V \geq V_{\mathit{max}}$ by time $\text{GST} + 3\delta$.
As $V_{\mathit{sync}}^* - V_{\mathit{max}} = O(\sqrt{n})$, each correct process sends $O(1) \cdot O(n\kappa) + O(\sqrt{n}) \cdot O(n\kappa) = O(n^{1.5}\kappa)$ bits after GST and before $\tau_{\mathit{sync}}^* + \Delta$.
Hence, \sync exchanges $O(n^{2.5}\kappa)$ bits.

As for the latency, each correct process enters $V_{\mathit{max}} + 1$ by time $\tau_1 = \text{GST} + 2\delta + \mathit{view\_duration} + 2\delta$: after $V_{\mathit{max}}$ is concluded (at time $\text{GST} + 2\delta + \mathit{view\_duration}$, at the latest), every correct process receives $2t + 1$ \textsc{view-completed} messages (line~\ref{line:receive_view_completed}) for $V_{\mathit{max}}$ by time $\text{GST} + 2\delta + \mathit{view\_duration} + \delta$, and enters $V_{\mathit{max}} + 1$ (line~\ref{line:enter_v}) by time $\text{GST} + 2\delta + \mathit{view\_duration} + 2\delta$.
Then, for each view $V > V_{\mathit{max}} + 1$, $V$ is entered by all correct processes by time $\tau_1 + (V - V_{\mathit{max}} - 1)(\mathit{view\_duration} + 2\delta)$.
Hence, $V_{\mathit{sync}}^*$ is entered by time $\tau_{\mathit{sync}}^* =\text{GST} + 2\delta + \mathit{view\_duration} + 2\delta + O(\sqrt{n})(\mathit{view\_duration} + 2\delta)$.
Therefore, $\tau_{\mathit{sync}}^* + \Delta - \text{GST} = O(\sqrt{n}) +  \mathit{view\_duration} + O(\sqrt{n})\mathit{view\_duration} = O(n)$.
\end{proof}

\subsection{Proof of Correctness \& Complexity} \label{subsection:spread_proof}

This subsection proves the correctness and complexity of \spread.
We start by proving that \spread satisfies integrity.

\begin{lemma} \label{lemma:spread_integrity}
\spread satisfies integrity.
\end{lemma}
\begin{proof}
Before a correct process acquires a hash-signature pair $(h, \Sigma_h)$ (line~\ref{line:acquire_spread}), the process checks that $\mathsf{verify\_sig}(h, \Sigma_h) = \mathit{true}$ (line~\ref{line:receive_confirm}).
Thus, \spread satisfies integrity.
\end{proof}

Next, we prove redundancy of \spread.

\begin{lemma} \label{lemma:spread_redundancy}
\spread satisfies redundancy.    
\end{lemma}
\begin{proof}
Let a correct process obtain a threshold signature $\Sigma_h$ such that $\mathsf{verify\_sig}(h, \Sigma_h) = \mathit{true}$, for some hash value $h$.
As $\Sigma_h$ is a valid threshold signature for $h$, $2t + 1$ processes have partially signed $h$, which implies that at least $t + 1$ correct processes have done so (line~\ref{line:share_sign_hash}).
These processes have obtained a value $v$ for which $\mathsf{hash}(v) = h$ (line~\ref{line:obtain}), which concludes the proof.
\end{proof}

Next, we prove termination of \spread.
We start by showing that, if a correct process acquires a hash-signature pair, all correct processes eventually acquire a hash-signature pair.

\begin{lemma} \label{lemma:uniformity}
If a correct process acquires a hash-signature pair at some time $\tau$, then every correct process acquires a hash-signature pair by time $\max(\text{{\emph{GST}}},\tau)+\delta$.
\end{lemma}
\begin{proof}
Let $P_i$ be a correct process that acquires a hash-signature pair at time $\tau$ (line~\ref{line:acquire_spread}). 
Immediately after acquiring the pair, $P_i$ disseminates the pair to all processes (line~\ref{line:echo_confirm}).
Hence, every correct process $P_j$ receives the disseminated pair (via a \textsc{confirm} message) by time $\max(\text{GST}, \tau) + \delta$.
If $P_j$ has not previously received a \textsc{confirm} message and has not stopped participating in \spread (by executing line~\ref{line:stop_executing_spread}), $P_j$ acquires the pair (line~\ref{line:acquire_spread}).
Otherwise, $P_j$ has already acquired a hash-signature pair (by time $\max(\text{GST}, \tau) + \delta$).
\end{proof}

The following lemma proves that, if all correct processes overlap for ``long enough'' in a view with a correct leader after GST, they all acquire a hash-signature pair.

\begin{lemma}\label{lemma:termination_under_synchronization}
If (1) all correct processes are in the same view $V$ from some time $\tau \geq \text{GST}$ until time $\tau' = \tau + \Delta = \tau + \delta\sqrt{n} + 3\delta$, and (2) $V$ has a correct leader, then all correct processes acquire a hash-signature pair by time $\tau'$ (and terminate \spread).
\end{lemma}
\begin{proof}
By the statement of the lemma, the view $V$ has a correct leader $P_l$.
$P_l$ sends a \textsc{dispersal} message (line~\ref{line:multicast_dispersal}) to every correct process by time $\tau + \delta\sqrt{n}$.
Thus, by time $\tau + \delta\sqrt{n} + \delta$, every correct process has received a \textsc{dispersal} message from $P_l$ (line~\ref{line:receive_dispersal}), and sent an \textsc{ack} message back (line~\ref{line:send_ack}).
Hence, $P_l$ receives $2t + 1$ \textsc{ack} messages (line~\ref{line:receive_ack}) by time $\tau + \delta\sqrt{n} + 2\delta$, which implies that $P_l$ sends a \textsc{confirm} message to all processes by this time.
Finally, all correct processes receive the aforementioned \textsc{confirm} message (line~\ref{line:receive_confirm}), and acquire a hash-signature pair (line~\ref{line:acquire_spread}) by time $\tau' = \tau + \Delta = \tau + \delta\sqrt{n} + 3\delta$.
\end{proof}

We are now ready to prove termination of \spread.

\begin{lemma} \label{lemma:spread_termination}
\spread satisfies termination.    
\end{lemma}
\begin{proof}
If a correct process acquires a hash-signature pair, termination is ensured by Lemma~\ref{lemma:uniformity}.
Otherwise, \sync ensures that there exists a synchronization time $\tau \geq \text{GST}$ such that (1) all correct processes are in the same view $V$ from $\tau$ until $\tau + \Delta = \tau + \delta\sqrt{n} + 3\delta$, and (2) $V$ has a correct leader.
In this case, termination is guaranteed by Lemma~\ref{lemma:termination_under_synchronization}.
\end{proof}

As it is proven that \spread satisfies integrity (Lemma~\ref{lemma:spread_integrity}), redundancy (Lemma~\ref{lemma:spread_redundancy}) and termination (Lemma~\ref{lemma:spread_termination}), \spread is correct.

\begin{theorem}
\spread is correct.
\end{theorem}

Now, we address \spread's complexity.
We start by proving that (at most) $O(n^{1.5})$ \textsc{dispersal} messages are sent after GST.



\begin{lemma}\label{lemma:dispersal_complexity}
At most $O(n^{1.5})$ \textsc{dispersal} messages are sent by correct processes after GST.
\end{lemma}
\begin{proof}
During the time period $[\text{GST}, \text{GST} + 3\delta)$, each correct process can send \textsc{dispersal} messages (line~\ref{line:multicast_dispersal}) to only $O(1)$ groups of $\sqrt{n}$ processes due to the $\delta$-waiting step (line~\ref{line:wait_spread}).
Moreover, Lemma~\ref{lemma:termination_under_synchronization} claims that all correct processes stop executing \spread by the end of $V_{\mathit{sync}}^*$ (i.e., by time $\tau_{\mathit{sync}}^* + \Delta$).
Furthermore, the stabilization property of \sync guarantees that each correct process enters a view $V \geq V_{\mathit{max}}$ by time $\text{GST} + 3\delta$.
Similarly, the overlapping property of \sync states that, for any view greater than $V_{\mathit{max}}$, all correct processes overlap in that view for (at least) $\Delta$ time.

Due to the definition of $V_{\mathit{sync}}^*$ and $V_{\mathit{max}}$, no correct leader exists in any view $V$ with $V_{\mathit{max}} < V < V_{\mathit{sync}}^*$.
During the time period $[\text{GST} + 3\delta, \tau_{\mathit{sync}}^* + \Delta]$, only the correct leaders of $V_{\mathit{max}}$ (if any) and $V_{\mathit{sync}}^*$ send \textsc{dispersal} messages.
Hence, the number of sent \textsc{dispersal} messages is
\begin{equation*}
    \underbrace{O(n^{1.5})}_{[\text{GST}, \text{GST} + 3\delta)} + \underbrace{O(n^{1.5})}_{V_{\mathit{max}}} + \underbrace{O(n^{1.5})}_{V_{\mathit{sync}}^*} = O(n^{1.5}).
\end{equation*}
Thus, the lemma holds.

\end{proof}

Finally, we are ready to prove the complexity of \spread.

\begin{theorem}
\spread exchanges $O(n^{1.5}L + n^{2.5}\kappa)$ bits after GST (and before it terminates), and terminates in $O(n)$ time after GST.
\end{theorem}
\begin{proof}
All correct processes stop executing \spread by the end of $V_{\mathit{sync}}^*$ (i.e., by time $\tau_{\mathit{sync}}^* + \Delta$) due to Lemma~\ref{lemma:termination_under_synchronization}.
Due to the limited entrance property of \sync, each correct process enters $O(1)$ views during the time period $[\text{GST}, \text{GST} + 3\delta)$. 
Each correct process sends $O(\sqrt{n})$ \textsc{ack} messages (line~\ref{line:send_ack}) in each view.
Similarly, each correct leader sends $O(n)$ \textsc{confirm} messages (line~\ref{line:broadcast_confirm}) in a view.
Thus, during the time period $[\text{GST}, \text{GST} + 3\delta)$, each correct process sends $O(1) \cdot \big( O(\sqrt{n}\kappa) + O(n\kappa) \big) = O(n\kappa)$ bits through \textsc{ack} and \textsc{confirm} messages.

Due to the stabilization property of \sync, each correct process enters (at least) $V_{\mathit{max}}$ by time $\text{GST} + 3\delta$.
For each view $V$ with $V_{\mathit{max}} < V < V_{\mathit{sync}}^*$, no correct leader is associated with $V$ (due to the overlapping property of \sync  and definition of $V_{\mathit{sync}}^*$): each correct process sends $O(\sqrt{n}\kappa)$ bits (exclusively via \textsc{ack} messages) in any view $V$ with $V_{\mathit{max}} < V < V_{\mathit{sync}}^*$.
Hence, during the time period $[\text{GST} + 3\delta, \tau_{\mathit{sync}}^* + \Delta]$, each correct process sends (via \textsc{ack} and \textsc{confirm} messages)
\begin{equation*}
        \underbrace{O(\sqrt{n})}_\text{$V_{\mathit{sync}}^* - V_{\mathit{max}}$} \cdot \underbrace{O(\sqrt{n}\kappa)}_\text{\textsc{ack}} + \underbrace{O(\sqrt{n}\kappa)}_\text{\textsc{ack} in $V_{\mathit{max}}$ and $V_{\mathit{sync}}^*$} + \underbrace{O(n\kappa)}_\text{\textsc{confirm} in $V_{\mathit{max}}$ and $V_{\mathit{sync}}^*$} = O(n\kappa) \text{ bits.}
\end{equation*}
Furthermore, at most $O(n^{1.5})$ \textsc{dispersal} messages are sent after GST (by Lemma~\ref{lemma:dispersal_complexity}), which implies that $O(n^{1.5}L)$ bits are sent through \textsc{dispersal} messages.
Next, \sync itself exchanges $O(n^{2.5}\kappa)$ bits.
Hence, \spread exchanges 
\begin{equation*}
    \underbrace{O(n^{2.5}\kappa)}_\text{\sync} + \underbrace{O(n^{1.5}L)}_\text{\textsc{dispersal}} + \underbrace{n \cdot O(n\kappa)}_\text{\textsc{ack} \& \textsc{confirm}} = O(n^{1.5}L + n^{2.5}\kappa) \text{ bits after GST and before termination}.
\end{equation*}
Note that the equation above neglects \textsc{confirm} messages sent by correct processes at line~\ref{line:echo_confirm}. 
However, as each correct process broadcasts this message only once (for a total of $O(n^2\kappa)$ sent bits by all correct processes), the complexity of \spread is not affected.
The latency of \spread is $\tau_{\mathit{sync}}^* + \Delta - \text{GST} = O(n)$ (due to the complexity of \sync).
\end{proof}
\section{\darestark: Proof}
\label{section:darestark_proof}

We start by proving agreement of \darestark.

\begin{lemma} \label{lemma:stark_agreement}
\darestark satisfies agreement.    
\end{lemma}
\begin{proof}
By contradiction, suppose that \darestark does not satisfy agreement.
Hence, there exists an execution in which two correct processes $P_i$ and $P_j$ disagree; let $P_i$ decide $v_i$ and $P_j$ decide $v_j \neq v_i$.
Process $P_i$ (resp., $P_j$) has received $t + 1$ RS symbols (line~\ref{line:stark_decide_rule}) which correspond to some hash value $h_i$ (resp., $h_j$).
Thus, one correct process has broadcast a \textsc{retrieve} message for $h_i$, and one correct process has broadcast a \textsc{retrieve} message for $h_j$ (line~\ref{line:stark_broadcast_retrieval}).
As correct processes broadcast \textsc{retrieve} messages only for hash values decided by \q (line~\ref{line:stark_q_decide}), $h_i = h_j = h$.
Moreover, $\mathit{decision\_symbols}_i[h]$ and $\mathit{decision\_symbols}_j[h]$ contain $t + 1$ correct RS symbols (due to the check at line~\ref{line:stark_check_sample}).
Therefore, \Cref{theorem:starkshard} shows that $\mathsf{hash}(v_i) = \mathsf{hash}(v_j) = h$, which implies that $v_i = v_j$ (as $\mathsf{hash}(\cdot)$ is collision-resistant).
We reach a contradiction, which concludes the proof.
\end{proof}

Next, we prove the validity property of \darestark.

\begin{lemma} \label{lemma:stark_validity}
\darestark satisfies validity.
\end{lemma}
\begin{proof}
Let a correct process $P_i$ decide a value $v$ (line~\ref{line:stark_decide}).
The value $v$ is ``constructed'' from $\mathit{decision\_symbols}_i$, which contains $t + 1$ correct RS symbols (due to the check at line~\ref{line:stark_check_sample}).
By \Cref{theorem:starkshard}, $\mathsf{valid}(v) = \mathit{true}$.
\end{proof}

The next lemma proves termination.

\begin{lemma} \label{lemma:stark_termination}
\darestark satisfies termination.
\end{lemma}
\begin{proof}
Each correct process $P_i$ successfully disperses its value's symbols (line~\ref{line:stark_send_dispersal}), each with a valid STARK proof.
When a correct process receives its symbol from $P_i$, the process eventually verifies the received proof (line~\ref{line:stark_receive_shard}), and partially signs the hash of $P_i$'s value (line~\ref{line:stark_send_ack}).
Therefore, $P_i$ eventually receives $2t + 1$ partial signatures (line~\ref{line:stark_received_acks}), and proposes to \q (line~\ref{line:stark_agreement_propose}).
Due to the termination and agreement properties of \q, all correct processes decide the same hash-signature pair $(h, \Sigma)$ from \q (line~\ref{line:stark_q_decide}).

At least $t + 1$ correct processes have obtained (provably correct) RS symbols for a value $v$ with $\mathsf{hash}(v) = h$ (as $2t + 1$ processes have partially signed $h$ at line~\ref{line:stark_send_ack}).
All of these correct processes disseminate their symbols to all other correct processes (line~\ref{line:stark_broadcast_retrieval}).
Hence, all correct processes eventually obtain $t + 1$ correct RS symbols (line~\ref{line:stark_decide_rule}), and decide (line~\ref{line:stark_decide}).
\end{proof}

As \darestark satisfies agreement (Lemma~\ref{lemma:stark_agreement}), validity (Lemma~\ref{lemma:stark_validity}) and termination (Lemma~\ref{lemma:stark_termination}), \darestark is correct.

\begin{theorem}
\darestark is correct.
\end{theorem}

Finally, we prove \darestark's complexity.

\begin{theorem}
\darestark achieves $O(nL + n^2\mathit{poly}(\kappa))$ bit complexity and $O(n)$ latency.
\end{theorem}
\begin{proof}
Recall that the bit complexity of \q is $O(n^2\kappa)$.
Moreover, each process sends (1) $O(n)$ \textsc{dispersal} messages (line~\ref{line:stark_send_dispersal}), each with $O(\frac{L}{n} + \mathit{poly}(\kappa))$ bits, (2) $O(n)$ \textsc{ack} messages (line~\ref{line:stark_send_ack}), each with $O(\kappa)$ bits, and (3) $O(n)$ \textsc{retrieve} messages (line~\ref{line:stark_broadcast_retrieval}), each with $O(\frac{L}{n} + \mathit{poly}(\kappa))$ bits.
Hence, through \textsc{dispersal}, \textsc{ack} and \textsc{retrieve} messages, each correct process sends 
\begin{equation*}
    \underbrace{n \cdot O(\frac{L}{n} + \mathit{poly}(\kappa))}_{\text{\textsc{dispersal}}} + \underbrace{n \cdot O(\kappa)}_{\text{\textsc{ack}}} + \underbrace{n \cdot O(\frac{L}{n} + \mathit{poly}(\kappa))}_{\text{\textsc{retrieve}}} = O(L + n\mathit{poly}(\kappa)) \text{ bits}.
\end{equation*}
Thus, the bit complexity of \darestark is $O(nL + n^2\mathit{poly}(\kappa))$.
Moreover, the latency of \darestark is $O(n)$ due to \q's linear latency.
\end{proof}
\section{Further Analysis of \name}

In this section, we provide a brief good-case analysis of \name and discuss how \name can be adapted to a model with unknown $\delta$.

\subsection{Good-Case Complexity}
\label{appendix:best_case}

For the good-case complexity, we consider only executions where $\text{GST} = 0$ and where all processes behave correctly.
This is sometimes also regarded as the common case since, in practice, there are usually no failures and the network behaves synchronously.
Throughout the entire subsection, $X = Y = \sqrt{n}$.

In such a scenario, the good-case bit complexity of \name is $O(n^{1.5}L + n^{2}\kappa)$.
As all processes are correct and synchronized at the starting view, \spread terminates after only one view.
The $n^{1.5}L$ term comes from this view: the first $\sqrt{n}$ correct leaders broadcast their full $L$-bit proposal to all other processes.
The $n^{2.5}\kappa$ term is reduced to only $n^{2}\kappa$ (only the \textsc{confirm} messages sent by correct processes at line~\ref{line:echo_confirm}) since \spread terminates after just one view.

The good-case latency of \name is essentially the sum of the good-case latencies of the Dispersal, Agreement, and Retrieval phases:

\begin{equation*}
        \begin{split}
        &\underbrace{O(\sqrt{n} \cdot \delta)}_\text{\textsc{dispersal}} + \underbrace{O(\delta)}_{\text{\textsc{agreement}}} + \underbrace{O(\delta)}_\text{\textsc{retrieval}} = O(\sqrt{n} \cdot \delta).
        \end{split}
    \end{equation*}

Thus, the good-case latency of \name is $O(\sqrt{n} \cdot \delta)$.

\subsection{\name (and \darestark) with Unknown \texorpdfstring{$\delta$}{Delta}} \label{appendix:dare_unknown}

To accommodate for unknown $\delta$, two modifications to \name are required:
\begin{compactitem}
    \item \spread must accommodate for unknown $\delta$.
    We can achieve this by having \sync increase the ensured overlap with every new view (by increasing $\mathit{view\_duration}$ for every new view).

    \item \q must accommodate for unknown $\delta$.
    Using the same strategy as for \sync, \q can tolerate unknown $\delta$.
    (The same modification makes \darestark resilient to unknown $\delta$.)
\end{compactitem}

\end{document}